\documentclass[11pt,letterpaper]{article}
\usepackage[margin=1in]{geometry}

\usepackage{amsmath,amsfonts,amssymb,amsthm,mathtools}
\usepackage{algorithm,algorithmic}
\usepackage{enumitem}
\usepackage{xspace,xcolor}
\usepackage{comment}
\usepackage{subcaption}
\usepackage{tikz}
\usetikzlibrary{positioning}
\tikzstyle{agent}=[circle,draw,minimum size=0.6cm,inner sep=0pt]
\tikzstyle{envyEdge}=[-latex,line width=0.03cm]
\tikzstyle{eqEdge}=[-latex,line width=0.03cm]
\tikzset{node distance=2cm,on grid}
\usepackage[normalem]{ulem}
\usepackage[square]{natbib} % Work with \bibliographystyle{plainnat}

\newtheorem{theorem}{Theorem}[section]
\newtheorem{lemma}[theorem]{Lemma}

\newtheorem{corollary}[theorem]{Corollary}

\theoremstyle{definition}
\newtheorem{definition}[theorem]{Definition}
\newtheorem{example}[theorem]{Example}

\newif\ifsubmission
\submissionfalse % full version
\newcommand{\longEFM}{envy-freeness for mixed goods\xspace}
\newcommand{\eferror}{\gamma\xspace}
\newcommand{\envyArrow}{\xrightarrow{\text{ENVY}}}
\newcommand{\eqArrow}{\xrightarrow{\text{EQ}}}
\newcommand{\epsEnvyArrow}{\xrightarrow{\epsilon\text{-ENVY}}}
\newcommand{\epsEqArrow}{\xrightarrow{\epsilon\text{-EQ}}}
\newcommand{\perfect}{\textsc{PerfectAlloc}\xspace}
\newcommand{\epsEF}[1]{#1\textsc{-EFAlloc}}

\title{Fair Division of Mixed Divisible and Indivisible Goods\thanks{A preliminary version appeared in Proceedings of the 34th AAAI Conference on Artificial Intelligence (AAAI)~\citep{BeiLiLi20}.
Compared to the conference version, this journal version fixes a bug in the proof of Theorem~\ref{thm:existenceOfEFhalf} and includes a new section (Section~\ref{sec:EFM_PO}) that discusses how to combine the newly proposed fairness notion together with economic efficiency considerations.}}
\author{
	Xiaohui Bei\\
	Nanyang Technological University\\
	\texttt{xhbei@ntu.edu.sg}\\
	\and
	Zihao Li\\
	Nanyang Technological University\\
	\texttt{zihao004@e.ntu.edu.sg}\\
	\and
	Jinyan Liu\\
	Beijing Institute of Technology\\
	\texttt{jyliu@cs.hku.hk}\\
	\and
	Shengxin Liu\\
	Nanyang Technological University\\
	\texttt{sxliu@ntu.edu.sg}\\
	\and
	Xinhang Lu\\
	Nanyang Technological University\\
	\texttt{xinhang001@e.ntu.edu.sg}\\
}

\date{}

\begin{document}
\maketitle

\begin{abstract}
We study the problem of fair division when the set of resources contains both divisible and indivisible goods.
Classic fairness notions such as envy-freeness (EF) and envy-freeness up to one good (EF1) cannot be directly applied to this mixed goods setting.
In this work, we propose a new fairness notion, \emph{\longEFM (EFM)}, which is a direct generalization of both EF and EF1 to the mixed goods setting.
We prove that an EFM allocation always exists for any number of agents with additive valuations.
We also propose efficient algorithms to compute an EFM allocation for two agents with general {additive} valuations and for $n$ agents with piecewise linear valuations over the divisible goods.
Finally, we relax the {envy-freeness} requirement, instead asking for \emph{$\epsilon$-\longEFM ($\epsilon$-EFM)}, and present an efficient algorithm that finds an $\epsilon$-EFM allocation.
%in time polynomial in the number of agents, the number of indivisible goods, and $1/\epsilon$. %% Reviewer 3 suggested simply writing ``efficient algorithm''.
\end{abstract}
% \begin{keyword}
% Fair division, Resource allocation, Envy-freeness, Social choice.
% \end{keyword}

\newpage

\section{Introduction}
Fair division studies the allocation of scarce resources among interested agents, with the objective of finding an allocation that is fair to all participants involved.
Initiated by~\citet{Steinhaus48}, the study of fair division has since been attracting interest from various disciplines for decades, including among others, mathematics, economics, and computer science~\citep{BramsTa96,RobertsonWe98,Moulin03,Thomson16,Moulin19}.

The literature of fair division can be divided into two classes, categorized by the type of the resources to be allocated.
The first class assumes the resource to be heterogeneous and infinitely divisible.
The corresponding problem is commonly known as \emph{cake cutting}.
One of the most prominent fairness notions in this setting is \emph{envy-freeness (EF)}.
An allocation is said to be envy-free if each agent {weakly} prefers her own bundle to any other bundle in the allocation.
An envy-free allocation with divisible resources always exists~\citep{Alon87,Su99} and can be found via a discrete and bounded protocol~\citep{AzizMa16}.

The second class considers the fair allocation of \emph{indivisible} goods.
Note that {an envy-free allocation} may fail to exist in the indivisible goods setting.\footnote{Consider the case where there are two agents but only a single valuable good to be allocated.}
To circumvent this problem, relaxations of envy-freeness have been studied.
One of the commonly considered relaxations is \emph{envy-freeness up to one good (EF1)}~\citep{LiptonMaMo04,Budish11}.
An allocation is said to satisfy EF1 if it is possible to eliminate any envy one agent has towards another agent by removing some good from the latter's bundle.
%no agent prefers the bundle of another agent to her own bundle following the removal of some good in the former bundle.
An EF1 allocation with indivisible goods always exists and can be found in polynomial time~\citep{LiptonMaMo04,CaragiannisKuMo19}.

The vast majority of the fair division literature assumes that the resources either are completely divisible, or consist of only indivisible goods.
However, this is not always the case in many real-world scenarios.
In inheritance division, for example, the inheritances to be divided among the heirs may contain divisible goods such as land and money, as well as indivisible goods such as houses, cars, and artworks.
What fairness notion should one adopt when dividing such mixed type of resources?
While EF and EF1 both work well in their respective settings, neither of them can be directly applied to this more general scenario. On the one hand, an EF allocation may not exist, when, for example, all goods are indivisible.
%On the other hand, an EF1 allocation may result in extreme unfairness when most resources are divisible. Consider the example where there is only one divisible good to be allocated. Giving the whole good to a single agent is EF1 but is obviously very unfair.
On the other hand, the EF1 notion in the mixed goods setting, when interpreted as that each agent does not envy another agent after removing at most one indivisible good from the latter agent's bundle, may also produce unfair allocations.
Consider the example where there is an indivisible good and a cake that are both equally valued by two agents.
The allocation that divides the cake in half and then {gives} the indivisible good to one of the agents is EF1  but is arguably unfair.
Another tempting solution is to divide the divisible and indivisible resources using EF and EF1 protocols separately and independently, and then combine the two allocations together.
This approach, however, also has problems.
Consider a simple example where two agents need to divide a cake \emph{and} an indivisible item. EF1 requires to allocate the indivisible item to one of the agent, say agent 1 for example.
However, if we then divide the cake using an arbitrary EF allocation, the overall allocation might be unfair to agent 2 who does not receive the indivisible item.
In fact, if the whole cake is valued less than the item, it would make more sense to allocate the cake entirely to agent 2.
When the cake is valued more than the item, it is still a fairer solution to allocate more cake to agent 2 in order to compensate her disadvantage in the indivisible resource allocation.
This demonstrates that it is not straightforward to generalize EF and EF1 to the mixed goods setting.
Dividing mixed types of resources calls for a new fairness notion that could unify EF and EF1 together to the new setting in a natural and non-trivial way.

% In a word, we need give a second thought on the fairness notions in the mixed goods setting. A desired fairness notion, whose existence is guaranteed in the mixed goods setting as well as is a generalization of both EF and EF1, is intensely needed.

\subsection{Our Results}
In this work, we initiate the study of fair division with mixed types of resources.
More specifically, we propose a new fairness notion, denoted as \emph{\longEFM} (or \emph{EFM} for short), that naturally combines EF and EF1 together and works for the setting where the set of resources may contain both divisible and indivisible goods.
Intuitively, EFM requires that for each agent, if her allocation consists of only indivisible items, then {other agents} will compare their bundles to hers using the EF1 criterion; but if this agent's bundle contains \emph{any} positive amount of divisible resources, {other agents} will compare their bundles to hers using the stricter EF condition.
This definition generalizes both EF and EF1 to the mixed goods setting and strikes a natural balance between the two fairness notions.
% When all goods are divisible (resp. indivisible), EFM serves the role of EF (resp. EF1). In addition, EFM always exists. We also introduce \emph{$\epsilon$-\longEFM ($\epsilon$-EFM)}, an approximation of the EFM notion, and we provide a polynomial algorithm to compute such an allocation.

In Section~\ref{sec:EFMExistence}, we first show that with mixed types of goods, an EFM allocation always exists for any number of agents with additive valuations.
Our proof is constructive and gives an algorithm for computing such an EFM allocation.
The algorithm requires an oracle for computing a \emph{perfect allocation in cake cutting} and can compute an EFM allocation in a polynomial number of steps.
In addition, in Section~\ref{sec:EFMSpecial}, we present two algorithms that could compute an EFM allocation for two special cases without using the perfect allocation oracle: (1) two agents with general {additive} valuations in the Robertson-Webb model, and (2) any number of agents with piecewise linear valuation functions.

While it is still unclear to us whether in general an EFM allocation can be computed in {a finite number of steps} in the Robertson-Webb model, in Section~\ref{sec:approxEFM}, we turn our attention to approximations and define the notion of $\epsilon$-EFM.
We then give an algorithm to compute an $\epsilon$-EFM allocation in the Robertson-Webb model with running time \emph{polynomial} in the number of agents $n$, the number of indivisible goods $m$, and $1/\epsilon$, and query complexity \emph{polynomial} in $n$ and $1/\epsilon$.
{We note that this algorithm does not require a perfect allocation oracle.}
This is an appealing result in particular due to its polynomial running time complexity.
A bounded exact EFM protocol, even if exists, is likely to require a large number of queries and cuts. This is because in the special case when resources are all divisible, EFM reduces to EF in cake cutting, for which the best known protocol~\citep{AzizMa16} has a very high query complexity (a tower of exponents of $n$).
This result shows that if one is willing to allow a small margin of errors, such an allocation could be found much more efficiently.

Finally, in Section~\ref{sec:EFM_PO} we discuss EFM in conjunction with efficiency considerations.
In particular, to one's surprise, we show that EFM and Pareto optimality (PO) are incompatible.
We also propose a weaker version of EFM and discuss the possibilities and difficulties in combining it with PO.

\subsection{Related Work}
As we mentioned, most previous works in fair division are from two categories based on whether the resources to be allocated are divisible or indivisible.

%When the resources are divisible, the existence of an envy-free allocation is guaranteed~\citep{Alon87}, even with $n-1$ cuts~\citep{Su99}.
When the resources are divisible, the existence of an envy-free allocation is guaranteed~\citep{Liapounoff40,DubinsSp61},
% \footnote{{The existence of envy-free allocations dates back to~\citet{Liapounoff40}, but was first formulated in the cake cutting setting by~\citet{DubinsSp61}.}}
even with only $n-1$ cuts~\citep{Stromquist80,Su99}.
\citet{BramsTa95} gave the first finite (but unbounded) envy-free protocol for any number of agents.
Recently, \citet{AzizMa16-STOC} gave the first bounded protocol for computing an envy-free allocation with four agents and their follow-up work extended the result to any number of agents~\citep{AzizMa16}.
Besides envy-freeness, other classic fairness notions include \emph{proportionality} and \emph{equitability}, both of which have been studied extensively~\citep{DubinsSp61,EvenPa84,EdmondsPr06,CechlarovaPi12,ProcacciaWa17}.

When the resources are indivisible, none of the aforementioned fairness notions is guaranteed to exist, thus relaxations are considered.
Among other notions, these include envy-freeness up to one good (EF1), envy-freeness up to any good (EFX), maximin share (MMS), etc.~\citep{LiptonMaMo04,Budish11,CaragiannisKuMo19}.
An EF1 allocation always exists and can be efficiently computed~\citep{LiptonMaMo04,CaragiannisKuMo19}.
However, the existence of an EFX allocation is still open~\citep{Procaccia20}, except for several special cases~\citep{PlautRo20,ChaudhuryGaMe20,AmanatidisBiFi20}.
As for MMS, an MMS allocation may not always exist; however, an approximation of MMS always exists and can be efficiently computed~\citep{KurokawaPrWa18,AmanatidisMaNi17,GhodsiHaSe18,GargTa20}. A recent paper by \citet{BeiLiLu20} also studied the existence, approximation and computation of MMS allocations in the mixed goods setting.

In addition, several works studied fair division with the assumption that resources can be shared among agents.
The adjusted-winner (AW) procedure proposed by~\citet{BramsTa96} ensures that at most one good must be split in a fair and (economically) efficient division between two agents.
\citet{SandomirskiySe19} focused on obtaining a fair and efficient division with minimum number of objects shared between two or more agents.
% While the aforementioned works obtaining fairness by leveraging the assumption that \emph{all} items are divisible, a notable exception is the work of~\citet{Rubchinsky10}, who . Catering to the setting,
\citet{Rubchinsky10} considered the fair division problem between two agents with both divisible and indivisible items, and introduced three fairness notions with computationally efficient algorithms for finding them.
% Moreover, \citeauthor{Rubchinsky10} also provided existence conditions of proportional and equitable fair divisions.
All of the works discussed above assumed that divisible items are \emph{homogeneous}.

Several other works studied the allocation of both indivisible goods and money, with the goal of finding envy-free allocations~\citep{Maskin87,AlkanDeGa91,Klijn00,MeertensPoRe02,HalpernSh19,BrustleDiNa20}.
Money can be viewed as a homogeneous divisible good which is valued the same across all agents.
In our work, we consider a more general setting with heterogeneous divisible goods.
Moreover, these works focused on finding envy-free allocations with the help of {a} sufficient amount of money, which is again different from our goal in the sense that our method could also be used even in cases where the money is insufficient.

\section{Preliminaries}
We consider a resource allocation setting with both divisible and indivisible goods (mixed goods for short).
Denote by $N = \{1, 2, \dots, n\}$ the set of agents, $M = \{1, 2, \dots, m\}$ the set of indivisible goods, and $D = \{D_1, D_2, \ldots, D_\ell\}$ the set of $\ell$ heterogeneous divisible goods or \emph{cakes}.
Since the fairness notion we propose below does not distinguish pieces from different cakes, without loss of generality, we assume each cake $D_i$ is represented by the interval $[\frac{i-1}{\ell}, \frac{i}{\ell}]$,\footnote{We assume that agents' valuation functions over the cakes are non-atomic. Thus we can view two consecutive cakes as disjoint even if they intersect at one boundary point.} and use a single cake $C = [0, 1]$ to represent the union of all cakes.\footnote{Sometimes we will use an arbitrary interval $[a, b]$ to denote the resource for simplicity; this can be easily normalized back to $[0, 1]$.}

Each agent $i$ has a non-negative utility $u_i(g)$ for each indivisible good $g \in M$.
Agents' utilities for subsets of indivisible goods are additive, meaning that $u_i(M') = \sum_{g \in M'} u_i(g)$ for each agent $i$ and subset of goods $M' \subseteq M$.
Each agent $i$ also has a density function $f_{i} \colon [0, 1] \to \mathbb{R}^+ \cup \{0\}$, which captures how the agent values different parts of the cake.
The value of agent $i$ over a finite union of intervals $S \subseteq [0, 1]$ is defined as $u_i(S) = \int_{S} f_{i}\ dx$.

Denote by $\mathcal{M} = (M_1, M_2, \dots, M_n)$ the partition of $M$ into bundles such that agent $i$ receives bundle $M_i$.
Denote by $\mathcal{C} = (C_1, C_2, \dots, C_n)$ the division of cake $C$ such that $C_i \cap C_j = \emptyset$ and agent $i$ receives $C_i$, a union of finitely many intervals.
An \emph{allocation} of the mixed goods is defined as $\mathcal{A} = (A_1, A_2, \dots, A_n)$ where $A_i = M_i \cup C_i$ is the \emph{bundle} allocated to agent $i$.
Agent $i$'s utility for the allocation is then defined as $u_i(A_i) = u_i(M_i) + u_i(C_i)$.
We assume without loss of generality that agents' utilities are normalized to 1, i.e., $u_i(M \cup C) = 1$ for all $i \in N$.

Next, we define the fairness notions used in this paper.

\begin{definition}[EF]
An allocation $\mathcal{A}$ is said to satisfy \emph{envy-freeness (EF)} if for any agents $i, j \in N$, $u_i(A_i) \geq u_i(A_j)$.
\end{definition}

\begin{definition}[EF1]
With indivisible goods, an allocation $\mathcal{A}$ is said to satisfy \emph{envy-freeness up to one good (EF1)} if for any agents $i, j \in N$ where $A_j \neq \emptyset$, there exists $g \in A_j$ such that $u_i(A_i) \geq u_i(A_j \setminus \{g\})$.
\end{definition}

Neither EF nor EF1 alone is a suitable definition for mixed goods.
In this paper we introduce the following new fairness notion.
% On the one hand, EF allocation may not exist if the mixed goods consist of only the indivisible goods. On the other hand, EF1 allocation may not be fair enough if the mixed goods consist of only the divisible goods (e.g., the allocation of a whole cake to a single agent).
% So an appropriate fairness notion, whose existence is guaranteed in the mixed goods setting as well as is a generalization of both EF and EF1, is intensely needed.

% We are now ready to introduce the fairness notion that we consider in this paper.

% (strong) EFM
\begin{definition}[EFM]\label{def:EFhalf}
An allocation $\mathcal{A}$ is said to satisfy \emph{\longEFM (EFM)} if for any agents $i, j \in N$,
\begin{itemize}
\item if agent $j$'s bundle consists of only indivisible goods, there exists $g \in A_j$ such that $u_i(A_i) \geq u_i(A_j \setminus \{g\})$;
\item otherwise, $u_i(A_i) \geq u_i(A_j)$.
\end{itemize}
\end{definition}

%The intuition behind EFM is that any envy that an agent $i$ has towards another agent $j$ may be eliminated by removing some resources from $j$'s bundle.
% The intuition behind EFM is that if agent $i$ envies agent $j$ and there is some divisible good in $j$'s bundle, it is possible to slightly move some divisible good from $j$'s bundle to $i$'s so that the envy of $i$ towards $j$ would be reduced.
% We can repeat this thought experiment till the moment that $i$ no longer envies $j$, or $j$'s bundle does not contain any divisible good.
It is easy to see that when the goods are all divisible, EFM reduces to EF; when goods are all indivisible, EFM reduces to EF1.
Therefore EFM is a natural generalization of both EF and EF1 to the mixed goods setting.

Next, we define $\epsilon$-EFM which is a relaxation of EFM.
Note that this definition only relaxes the EF condition for the divisible goods; the EF1 condition is not relaxed.

\begin{definition}[$\epsilon$-EFM]\label{def:epsEFhalf}
An allocation $\mathcal{A}$ is said to satisfy \emph{$\epsilon$-\longEFM ($\epsilon$-EFM)} if for any agents $i, j \in N$,
\begin{itemize}
\item if agent $j$'s bundle consists of only indivisible goods, there exists $g \in A_j$ such that $u_i(A_i) \geq u_i(A_j \setminus \{g\})$;
\item otherwise, $u_i(A_i) \geq u_i(A_j) - \epsilon$.
\end{itemize}
\end{definition}

Finally, we describe the Robertson-Webb (RW) query model~\citep{RobertsonWe98}, which is a standard model in cake cutting. In this model, an algorithm is allowed to interact with the agents via two types of queries:
\begin{itemize}
\item Evaluation: An evaluation query of agent $i$ on $[x, y]$ returns $u_i([x, y])$.
\item Cut: A cut query of $\beta$ for agent $i$ from $x$ returns a point $y$ such that $u_i([x, y]) = \beta$.
\end{itemize}

%In this paper, we assume each query in the RW model takes unit time.

\section{EFM: Existence}\label{sec:EFMExistence}
Although EFM is a natural generalization of both EF and EF1, it is not straightforward whether an EFM allocation would always exist with mixed goods.
In this section, we prove through a constructive algorithm that with mixed goods and any number of agents, an EFM allocation always exists.

We first give some definitions which will be helpful for our algorithm and proofs.

\paragraph{Perfect Allocation}
Our algorithm will utilize the concept of \emph{perfect allocation} in cake cutting.

\begin{definition}[Perfect allocation]
A partition $\mathcal{C} = (C_1, C_2, \dots, C_k)$ of cake $C$ is said to be \emph{perfect} if for all $i \in N, j \in [k]$, $u_i(C_j) = u_i(C)/k$.
\end{definition}

Intuitively, a perfect allocation in cake cutting divides the cake into $k$ pieces, such that every agent in $N$ values these $k$ pieces equally.
It is known that a perfect allocation always exists for any number of agents and any $k$~\citep{Alon87}.
In the following, we will assume that our algorithm is equipped with an oracle $\perfect(C, k, N)$ that could return us a perfect allocation for any $k$ and cake $C$ among all agents in $N$.

% \begin{lemma}[\citet{Alon87}]\label{lemma:AlonPerfect}
% For a cake with $u_i(C)=1$ for all agents, it is possible to cut the cake into $(k-1) \cdot n$ pieces and partition these pieces into $k$ bundles $C_1, C_2, \dots, C_k$, such that $u_i(C_j) = 1/k$, for all $i \in [n], j \in [k]$. This is called as perfect allocation.
% \end{lemma}

% If valuations of the cake for all agents are not the same, based on the above lemma, it is easy to get an allocation such that each agent views all bundles equally.
% Firstly, normalize every agent's valuation to be $\tilde{u}_i(C)=1$.
% Then we can get a perfect allocation over these normalized valuations by Lemma~\ref{lemma:AlonPerfect}, as $\tilde{u}_i(C_j) = \tilde{u}_i(C)/k$ for all $i \in [n]$ and $j,j' \in [k]$.
% Finally, project $\tilde{u}_i$ back to $u_i$, we will have $u_i(C_j) = u_i(C)/k$.

\paragraph{Envy Graph and Addable Set}
We also make use of the \emph{envy graph} to capture the envy relation among agents in an allocation.

\begin{definition}[Envy graph]
Given an allocation $\mathcal{A}$, its corresponding \emph{envy graph} $G = (N, E_{\text{envy}} \cup E_{\text{eq}})$ is a directed graph, where each vertex represents an agent, and $E_{\text{envy}}$ and $E_{\text{eq}}$ consist of the following two types of edges, respectively:
\begin{itemize}
\item Envy edge: $i \envyArrow j$ if $u_i(A_i) < u_i(A_j)$;
\item Equality edge: $i \eqArrow j$ if $u_i(A_i) = u_i(A_j)$.
\end{itemize}
\end{definition}

Moreover, a cycle in an envy graph is called an \emph{envy cycle} if it contains at least one envy edge.
{The concepts of envy edge and equality edge were also used in~\citep{Klijn00,LiptonMaMo04}.}

Given an envy graph, we then define another useful concept called \emph{addable set} which corresponds to a specific set of agents.

\begin{definition}[Addable set]\label{def:addableSet}
Given an envy graph, a non-empty set of agents $S \subseteq N$ forms an \emph{addable set} if,
\begin{itemize}
\item there is no envy edge between any pair of agents in $S$;
\item there exists neither an envy edge nor an equality edge from any agent in $N \setminus S$ to any agent in $S$.
\end{itemize}
\end{definition}

Moreover, an addable set $S \subseteq N$ is called a \emph{maximal} addable set if there does not exist any other addable set $S' \subseteq N$ such that $S \subset S'$.
The following lemma shows the uniqueness of the maximal addable set in an envy graph.

\begin{lemma}\label{lem:max_addable}
Given an envy graph, the maximal addable set, if exists, is unique.
Moreover, we can find it or decide that none exists in $O(n^3)$ time.
\end{lemma}
\begin{proof}
Suppose, to the contrary, that there exist two distinct maximal addable sets $S_1$ and $S_2$ in the given envy graph.
We will show that $S_1 \cup S_2$ is also an addable set which contradicts the maximality of $S_1$ and $S_2$.

First it is easy to see that there exists neither an envy edge nor an equality edge from any agent in $N \setminus (S_1 \cup S_2)$ to agents in $S_1 \cup S_2$ since, otherwise, either $S_1$ or $S_2$ is not an addable set.

We next argue that there is no envy edge between any pair of agents in $S_1 \cup S_2$.
Clearly, according to Definition~\ref{def:addableSet}, there is no envy edge within each of $S_1$ and $S_2$.
The envy edges between $S_1$ and $S_2$ also cannot exist because there are no envy edges coming from outside of $S_1$ or $S_2$ into any of them.
Thus, $S_1 \cup S_2$ is also an addable set.

Now, we show how to find the unique maximal addable set or decide its non-existence in $O(n^3)$ time: for each $j$ which has an incoming envy edge, let $R_j$ be the collection of vertices (including $j$) that are reachable by $j$ via the union of envy edges and equality edges, and let $S = N \setminus \bigcup_j {R_j}$.
We will show that an addable set does not exist in the envy graph if $S = \emptyset$. Otherwise, $S$ is the unique maximal addable set.
First, $S$ is an addable set because any agent in $S$ does not have any incoming envy edge and is not reachable via the union of envy edges and equality edges from any other agent with an incoming envy edge.
In addition, $S$ is maximal because any agent in $\bigcup_j {R_j}$ cannot be in any addable set.
Such $S$ can be found in $O(n^3)$ time because it takes $O(n)$ time to check if an agent has an incoming edge, and for any agent $j$ who has an incoming envy edge, it then takes $O(n^2)$ time to construct $R_j$ via, for example, breadth-first search (BFS).
\end{proof}

Intuitively, agents in the addable set $S$ can be allocated some cake without creating new envy, since each agent in $N \setminus S$ values her own bundle strictly more than the bundles of agents in $S$.
Our next result characterizes the relation between the addable set and the envy cycle.

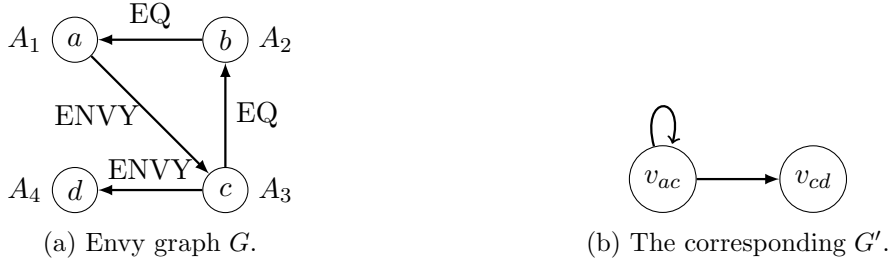
\begin{figure}[t]
\centering
\begin{subfigure}[t]{.45\linewidth}
\centering
\begin{tikzpicture}
\node[agent,label=left:{$A_1$}] (a) {$a$};
\node[agent,label=right:{$A_2$}] (b) [right=of a] {$b$};
\node[agent,label=right:{$A_3$}] (c) [below=of b] {$c$};
\node[agent,label=left:{$A_4$}] (d) [below=of a] {$d$};
\path[envyEdge] (a) edge node[left]{ENVY} (c);
\path[eqEdge] (b) edge node[above]{EQ} (a);
\path[envyEdge] (c) edge node[above]{ENVY} (d);
\path[eqEdge] (c) edge node[right]{EQ} (b);
\end{tikzpicture}
\subcaption{Envy graph $G$.}
\label{fig:G}
\end{subfigure}
~
\begin{subfigure}[t]{.45\linewidth}
\centering
\begin{tikzpicture}
\node[draw,circle] (vac) {$v_{ac}$};
\node[draw,circle] (vcd) [right=of vac] {$v_{cd}$};
\path[-latex,line width=0.03cm] (vac) edge (vcd);
\path[-latex,line width=0.03cm] (vac) edge[loop above] (vac);
\end{tikzpicture}
\subcaption{The corresponding $G'$.}
\label{fig:GPrime}
\end{subfigure}
\caption{Figure~\ref{fig:G} shows an envy graph $G$ with four vertices (agents) $a, b, c, d$.
The bundle each agent gets is labelled with $A_1, A_2, A_3, A_4$ beside. We show in Figure~\ref{fig:GPrime} its corresponding $G'$. In this example, $G$ has an envy cycle (involving vertices $a, b, c$) but no addable set.}
% We also give another example (Figure~\ref{fig:G-addableSet}) in which there is an addable set but no envy cycle.
\label{fig:G-envyCycle}
\end{figure}

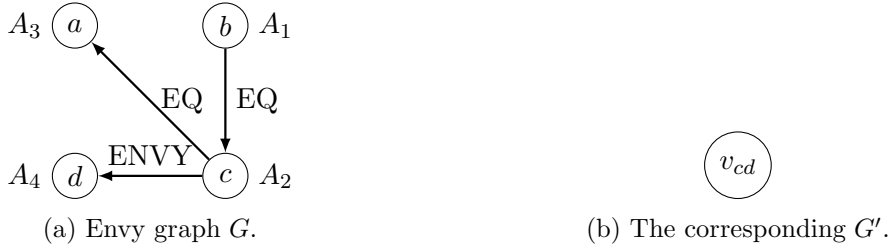
\begin{figure}[t]
\centering
\begin{subfigure}[t]{.45\linewidth}
\centering
\begin{tikzpicture}
\node[agent,label=left:{$A_3$}] (a) {$a$};
\node[agent,label=right:{$A_1$}] (b) [right=of a] {$b$};
\node[agent,label=right:{$A_2$}] (c) [below=of b] {$c$};
\node[agent,label=left:{$A_4$}] (d) [below=of a] {$d$};
\path[eqEdge] (b) edge node[right]{EQ} (c);
\path[envyEdge] (c) edge node[above]{ENVY} (d);
\path[eqEdge] (c) edge node[right]{EQ} (a);
\end{tikzpicture}
\subcaption{Envy graph $G$.}
\label{fig:GG}
\end{subfigure}
~
\begin{subfigure}[t]{.45\linewidth}
\centering
\begin{tikzpicture}
\node[draw,circle] (vcd) {$v_{cd}$};
\end{tikzpicture}
\subcaption{The corresponding $G'$.}
\label{fig:GGPrime}
\end{subfigure}
\caption{After rotating the bundles along the envy cycle in Figure~\ref{fig:G}, we obtain the envy graph in Figure~\ref{fig:GG}.
The corresponding $G'$ is shown in Figure~\ref{fig:GGPrime}.
In this example, $G$ has addable sets {$\{b\}$, $\{b, c\}$ and $\{a, b, c\}$} but no envy cycle.}
\label{fig:G-addableSet}
\end{figure}

\begin{lemma}\label{lem:addableSet-envyCycle}
Any envy graph $G = (N, E_{\text{envy}} \cup E_{\text{eq}})$ that does not have any envy cycle must have an addable set.
\end{lemma}
\begin{proof}
We assume without loss of generality that $E_{\text{envy}} \neq \emptyset$, since otherwise $N$ itself is an addable set.
Now, we construct graph $G' = (N', E')$ from $G$ as follows.
Each envy edge $i \envyArrow j$ in $G$ corresponds to a vertex $v_{ij}$ in $G'$.
For two envy edges $i \envyArrow j$ and $i' \envyArrow j'$ in $G$, if there exists a path from $j$ to $i'$, we construct an edge $v_{ij} \to v_{i'j'}$ in $G'$.
Note that, if there is an envy edge $i \envyArrow j$ and a path from $j$ to $i$ in $G$, there will be a self-loop $v_{ij} \to v_{ij}$ in $G'$.
We illustrate this transformation using two examples in Figures~\ref{fig:G-envyCycle} and~\ref{fig:G-addableSet}.

It is easy to see that a cycle in $G'$ implies an envy cycle in $G$.
Thus, by the assumption that there is no envy cycle in $G$, $G'$ must be acyclic.
Then there must exist a vertex $v_{ij} \in N'$ which is not reachable by any other vertices in $G'$.
Because $v_{ij}$ corresponds to the envy edge $i \envyArrow j$ in $G$,
since $v_{ij}$ cannot be reached by any vertices in $G'$, the vertex $i$ is also not reachable by any $j'$ which is pointed by an envy edge.
We note that, however, this vertex $i$ may be reachable by other vertices via only equality edges.
Thus, we need to not only include agent $i$ in the addable set but also those agents who are able to reach $i$ via equality edges.

Let $S$ be the set containing agent $i$ and all other agents who can reach $i$ in the envy graph $G$ via equality edges.
In the following, we show that $S$ is an addable set.
First, $S$ is non-empty because it at least contains agent $i$.
Second, by our construction, there is no envy edge between any pair of agents in $S$.
Third, recall that in envy graph $G$, agent $i$ is not reachable by any $j'$ which is pointed by an envy edge; thus, $S$ is also not pointed by any envy edge.
Last, $S$ is also not pointed by any equality edge by our construction of $S$.
Therefore, according to Definition~\ref{def:addableSet}, $S$ must be an addable set.
\end{proof}

\subsection{The Algorithm}
\begin{algorithm}[ht!]
\caption{EFM Algorithm}
\label{alg:EFhalf}
\begin{algorithmic}[1]
\REQUIRE Agents $N$, indivisible goods $M$ and cake $C$.
\STATE Find an arbitrary EF1 allocation $(A_1, A_2, \dots, A_n)$ of $M$ to $n$ agents. \label{EFhalfALG-EF1Allocation}
\STATE Construct an envy graph $G = (N, E_{\text{envy}} \cup E_{\text{eq}})$ accordingly. \label{EFhalfALG-initialEnvyGraph}

\WHILE {$C \neq \emptyset$} \label{EFhalfALG-remainingCake}
%	\IF {there exists an addable set $S$} \label{EFhalfALG-addCakeBegin}
	\IF {there exists an addable set in $G$} \label{EFhalfALG-addCakeBegin}
	\STATE // cake-adding phase
	\STATE Let $S$ be the maximal addable set. \label{EFhalfALG-maximalAddable}
		\IF {$S = N$} \label{EFhalfALG-EFbegin}
			\STATE Find an EF allocation $(C_1, C_2, \dots, C_n)$ of $C$. \label{EFhalfALG-EF}
			\STATE $C \leftarrow \emptyset$
			\STATE Add $C_i$ to bundle $A_i$ for all $i \in N$. \label{EFhalfALG-EFUpdate}
		\ELSE \label{EFhalfALG-SneqN}
			\STATE $\delta_i \leftarrow \min_{j \in S} (u_i(A_i) - u_i(A_j))$ for each $i \in N \setminus S$. \label{EFhalfALG-deltai}
			\IF {$u_i(C) \leq |S| \cdot \delta_i$ holds for each $i \in N \setminus S$} \label{EFhalfALG-chooseCakeBegin}
				\STATE $C' \leftarrow C, C \leftarrow \emptyset$
			\ELSE \label{EFhalfALG-cakeBegin}
				\STATE Suppose w.l.o.g.~that $C = [a, b]$. For each agent $i \in N \setminus S$, if $u_i([a, b]) \geq |S| \cdot \delta_i$, let $x_i$ be a point such that $u_i([a, x_i]) = |S| \cdot \delta_i$; otherwise, let $x_i = b$. \label{EFhalfALG-xi}
				\STATE $i^* \leftarrow \arg\min_{i \in N \setminus S} x_i$ \label{EFhalfALG-istar}
				\STATE $C' \leftarrow [a, x_{i^*}], C \leftarrow C \setminus C'$ \label{EFhalfALG-cakeEnd}
			\ENDIF \label{EFhalfALG-chooseCakeEnd}
			\STATE Let $(C_1, C_2, \dots, C_k) = \perfect(C', k, N)$ where $k = |S|$. \label{EFhalfALG-addCake}
			\STATE Add $C_{i}$ to the bundle of the $i$-th agent in $S$. \label{EFhalfALG-PerfectUpdate}
			\STATE Update envy graph $G$ accordingly.
		\ENDIF \label{EFhalfALG-addCakeEnd}
	\ELSE \label{EFhalfALG-envyCycleBegin}
		\STATE // envy-cycle-elimination phase
		\STATE Let $T$ be an envy cycle in envy graph $G$.
		\STATE For each agent $j \in T$, give agent $j$'s whole bundle to agent $i$ who points to her in $T$. \label{EFhalfALG-cycleElimination}
		\STATE Update envy graph $G$ accordingly.
	\ENDIF \label{EFhalfALG-envyCycleEnd}
\ENDWHILE \label{EFhalfALG-remainingEnd}

\RETURN $(A_1, A_2, \dots, A_n)$ \label{EFhalfALG-lastReturn}
\end{algorithmic}
\end{algorithm}

The complete algorithm to compute an EFM allocation is shown in Algorithm~\ref{alg:EFhalf}.

In general, our algorithm always maintains a partial allocation that is EFM.
Then, we repeatedly and carefully add resources to the partial allocation, until all resources are allocated.
We start with an EF1 allocation of only indivisible goods to all agents in Step~\ref{EFhalfALG-EF1Allocation}, and construct the corresponding envy graph in Step \ref{EFhalfALG-initialEnvyGraph}.
Then, our algorithm executes in rounds (Steps~\ref{EFhalfALG-remainingCake}-\ref{EFhalfALG-remainingEnd}).
In each round, we try to distribute some cake to the partial allocation while ensuring the partial allocation to be EFM.
Such distribution needs to be done carefully because once an agent is allocated with a positive amount of cake, the fairness condition with regard to her bundle changes from EF1 to EF, which is more demanding.
We repeat the process until the whole cake is allocated.

In each round of Algorithm~\ref{alg:EFhalf}, depending on whether there is an addable set that can be given some cake in Step~\ref{EFhalfALG-addCakeBegin}, we execute either the \emph{cake-adding phase} (Steps~\ref{EFhalfALG-addCakeBegin}-\ref{EFhalfALG-addCakeEnd})
or the \emph{envy-cycle-elimination phase} (Steps~\ref{EFhalfALG-envyCycleBegin}-\ref{EFhalfALG-envyCycleEnd}).

\begin{itemize}
\item In the cake-adding phase, we have a maximal addable set $S$.
By its definition, each agent in $N \setminus S$ values her own bundle strictly more than the bundles of agents in $S$.
Thus there is room to allocate some cake $C'$ to agents in $S$.
We carefully select $C'$ to be allocated to $S$ such that it does not create any new envy among the agents.
To achieve this, we choose a piece of cake $C' \subseteq C$ to be \emph{perfectly} allocated to $S$ in Steps~\ref{EFhalfALG-deltai}-\ref{EFhalfALG-chooseCakeEnd} so that no agent in $N$ will envy agents in $S$ after distributing $C'$ in Steps~\ref{EFhalfALG-addCake}-\ref{EFhalfALG-PerfectUpdate}.
More specifically, for each agent $i \in N \setminus S$, we determine in Step~\ref{EFhalfALG-deltai} the largest value $\delta_i$ to be added to any agent in $S$ such that $i$ would still not envy any agent in $S$.
Then, the way we decide $x_{i^*}$ in Steps~\ref{EFhalfALG-istar}-\ref{EFhalfALG-cakeEnd} ensures that for all agents $i \in N \setminus S$, $v_i([a, x_{i^*}]) \leq |S|\cdot \delta_i$.
Next, in Step~\ref{EFhalfALG-addCake}, cake $C' = [a, x_{i^*}]$ is divided into $|S|$ pieces that are valued equally by all agents in $N$. This is to ensure that no agent $i \in N \setminus S$ values any piece more than $\delta_i$.
% It implies that for all $i \in N \setminus S$, the value of cake $C'$ cannot exceed $|S| \cdot \delta_i$ (since we resort to perfect allocation later on).
% We determine in Step~\ref{EFhalfALG-xi} the point $x_i$ for each $i \in N \setminus S$ such that her value for the piece of cake $[a, x_i]$ exactly is $|S| \cdot \delta_i$, or let $x_i = b$ if $u_i([a, b]) < |S| \cdot \delta_i$.
% It is worth noting that as the cake is heterogeneous and each agent may have different values for different parts of the cake, $x_i$'s may not directly relate to $\delta_i$'s.
% Among these $x_i$'s, we choose the smallest one, denoted by $x_{i^*}$ (Steps~\ref{EFhalfALG-istar}-\ref{EFhalfALG-cakeEnd}).
% Then, in Step~\ref{EFhalfALG-addCake}, cake $C' = [a, x_{i^*}]$ is divided into $|S|$ pieces in the sense that the $|S|$ pieces are equal according to valuations of all $n$ agents.
% The value of cake $C'$ allocated to $S$ does not exceed $|S| \cdot x_i$ for all $i \in N \setminus S$, and is exactly $|S| \cdot x_{i^*}$ for agent $i^*$.}

\item In the envy-cycle-elimination phase, i.e., when there does not exist any addable set, we show that in this case there must exist an envy cycle $T$ in the current envy graph.
We can then apply the envy-cycle-elimination technique to reduce some existing envy from the allocation by rearranging the bundles along $T$.
More specifically, for each agent $j \in T$, we give agent $j$'s bundle to agent $i$ who points to her in $T$ (shown in Step~\ref{EFhalfALG-cycleElimination}).
\end{itemize}

We remark that when all goods are indivisible, our algorithm performs Steps~\ref{EFhalfALG-EF1Allocation}-\ref{EFhalfALG-initialEnvyGraph} and terminates with an EF1 allocation (which is also EFM).
When the whole good is a divisible cake, the algorithm goes directly to Step~\ref{EFhalfALG-EF} and ends with an EF allocation of the cake, which is again EFM.

In the following we prove the correctness of this algorithm and analyze its running time.

\subsection{Analysis}
Our main result for the EFM allocation is as follows:

\begin{theorem}\label{thm:existenceOfEFhalf}
An EFM allocation always exists for any number of agents with additive valuations and can be found by Algorithm~\ref{alg:EFhalf} in polynomial time with $O(n^4)$ Robertson-Webb queries and $O(n^3)$ calls to the \perfect oracle.
% \sout{An EFM allocation always exists for any number of agents and can be found by Algorithm~\ref{alg:EFhalf} with a perfect allocation oracle in polynomial time.}
\end{theorem}

To prove Theorem~\ref{thm:existenceOfEFhalf}, we first {show that the following invariants are} maintained by Algorithm~\ref{alg:EFhalf} during its run.

\vbox{
\hrule
\paragraph{Invariants}
\begin{enumerate}[label=\textbf{A\arabic*.},ref=A\arabic*]
\item In each round there is either an addable set for the cake-adding phase or an envy cycle for the envy-cycle-elimination phase. \label{item:addableSetOrEnvyCycle}
\item The partial allocation is always EFM. \label{item:partial-EFM}
\end{enumerate}
\hrule
}

\begin{lemma}\label{lem:envy-cycle-addable-set}
Invariant~\ref{item:addableSetOrEnvyCycle} holds during the algorithm's run.
\end{lemma}
\begin{proof}
This invariant is implied directly by Lemma~\ref{lem:addableSet-envyCycle}.
% If there is an envy cycle, we can perform the envy-cycle-elimination phase.
% Otherwise, according to Lemma~\ref{lem:addableSet-envyCycle}, there is an addable set, and hence we can perform the cake-adding phase.
\end{proof}

\begin{lemma}\label{lem:alwaysEFhalf}
Invariant~\ref{item:partial-EFM} holds during the algorithm's run.
\end{lemma}
\begin{proof}
The partial allocation is clearly EFM after Step~\ref{EFhalfALG-EF1Allocation}.
Then the allocation is updated in three places in the algorithm: Steps~\ref{EFhalfALG-EFUpdate} and~\ref{EFhalfALG-PerfectUpdate} in the cake-adding phase and Step~\ref{EFhalfALG-cycleElimination} in the envy-cycle-elimination phase.
Given a partial allocation that is EFM, we will show that each of these updates maintains the EFM condition.

First, when we have $S = N$ in Step~\ref{EFhalfALG-EFbegin}, i.e., the addable set $S$ consists of all $n$ agents, the current envy graph does not contain any envy edge due to the definition of addable set (Definition~\ref{def:addableSet}).
This implies that current partial allocation actually is envy-free.
Because all valuation functions are additive, adding another envy-free allocation on top of it in Step~\ref{EFhalfALG-EFUpdate} results in an envy-free and, hence, EFM allocation.

We next consider Step~\ref{EFhalfALG-PerfectUpdate} in the cake-adding phase where a piece of cake is added to the addable set $S$.
In order to maintain an EFM partial allocation, we need to ensure that this process does not introduce any new envy towards agents in $S$.
Since we add a perfect allocation in Steps~\ref{EFhalfALG-addCake}-\ref{EFhalfALG-PerfectUpdate}, envy will not emerge among agents in $S$.
We also carefully choose the amount of cake to be allocated in Steps \ref{EFhalfALG-chooseCakeBegin}-\ref{EFhalfALG-chooseCakeEnd} such that each agent in $N \setminus S$ weakly prefers her bundle to any bundles that belong to agents in $S$.
To achieve this, we choose a piece of cake $C' \subseteq C$ to be \emph{perfectly} allocated to $S$ in
Steps~\ref{EFhalfALG-deltai}-\ref{EFhalfALG-chooseCakeEnd} so that no agent in $N$ will envy agents in $S$ after distributing $C'$ in
Steps~\ref{EFhalfALG-addCake}-\ref{EFhalfALG-PerfectUpdate}.
More specifically, for each agent $i \in N \setminus S$, we determine in Step~\ref{EFhalfALG-deltai} the largest value $\delta_i$ to be added to any agent in $S$ such that $i$ would still not envy any agent in $S$.
Then, the way we decide $x_{i^*}$ in Steps~\ref{EFhalfALG-istar}-\ref{EFhalfALG-cakeEnd} ensures that for all agents $i \in N \setminus S$, $v_i([a, x_{i^*}]) \leq |S|\cdot \delta_i$.
Next, in Step~\ref{EFhalfALG-addCake}, cake $C' = [a, x_{i^*}]$ is divided into $|S|$ pieces that are valued equally by all agents in $N$. This is to ensure that for any piece of cake $C_j$ allocated to $j \in S$, we have $u_i(C_j) \leq \delta_i$ for all $i \in N \setminus S$.
Thus, agents in $N \setminus S$ continues to not envy agents in $S$ in Step~\ref{EFhalfALG-PerfectUpdate}.

Finally, in the envy-cycle-elimination phase, Step~\ref{EFhalfALG-cycleElimination} eliminates envy edges by rearranging the partial allocation within the envy cycle $T$.
Since each agent in $T$ is weakly better off, the partial allocation remains EFM.
For agents in $N \setminus T$, rearranging the partial allocation that is EFM will not make EFM infeasible.
The conclusion follows.
\end{proof}

\paragraph{Correctness}
\begin{lemma}\label{lem:EFhalfAllocation}
Algorithm \ref{alg:EFhalf} always returns an EFM allocation upon termination.
\end{lemma}
\begin{proof}
By Invariant~\ref{item:partial-EFM}, it suffices to prove that all goods are allocated when Algorithm \ref{alg:EFhalf} terminates.
All indivisible goods are allocated in Step \ref{EFhalfALG-EF1Allocation}.
Then the \texttt{while} loop (Steps \ref{EFhalfALG-remainingCake}-\ref{EFhalfALG-remainingEnd}) terminates only when the cake is also fully allocated, as desired.
\end{proof}

\paragraph{Termination and Time Complexity}
%We use the number of envy and equality edges in the envy graph as a potential function to bound the running time of this algorithm.

We use the number of envy edges in the envy graph and the size of {the} maximal addable set as a potential function to bound the running time of this algorithm.

\begin{lemma}\label{lem:equalEdgeIncrease}
After the algorithm completes a cake-adding phase, the number of envy edges never increases.
In addition, if the piece of cake to be allocated is not the whole remaining cake, either (a) the number of envy edges strictly decreases, or (b) the size of the maximal addable set strictly decreases {or an addable set no longer exists}.
\end{lemma}
\begin{proof}
By Lemma~\ref{lem:alwaysEFhalf}, the partial allocation is always EFM after a cake-adding phase.
In a cake-adding phase, some positive amount of cake is added to every agent in $S$. This means after this phase, there would never be any envy edge between agents in $S$ or from $N \setminus S$ to $S$.
The bundles of agents in $N \setminus S$ remains the same, hence the set of edges among agents in $N \setminus S$ remains unchanged.
Lastly, since only agents in $S$ are allocated new resources in the cake-adding phase, no new envy edge will be introduced from $S$ to $N \setminus S$.
This proves the first part of Lemma~\ref{lem:equalEdgeIncrease}.

For the second part, we only study the situation when the piece of cake to be allocated to agents in $S$ is not the whole remaining cake (Steps~\ref{EFhalfALG-cakeBegin}-\ref{EFhalfALG-cakeEnd}).
Note that the number of envy edges will never increase after a cake-adding phase as proved above.
It suffices to show that if the number of envy edges remains unchanged and an addable set still exists, then the size of the maximal addable set must strictly decrease.
% We first observed that, with the added cake, envy edges in Case (2) might become equality edges or totally disappear. If this happens, combined with the argument in Cases (1), (3) and (4), the number of envy edges will strictly decrease which corresponds to (a) in the lemma statement.

% Otherwise, i.e., the envy edges in Case (2) stay the same, we know that the set of edges in Cases (1) and (3) remains unchanged and the set of equality edges in Case (2) might be shrank. For Case (4),
Note that based on how we choose $i^*$ in Step~\ref{EFhalfALG-istar}, after the cake-adding phase, at least one equality edge will be generated in the envy graph from agent $i^*$ to some agent $j \in S$.
Let $G$ and $G'$ be the envy graphs before and after the cake-adding phase, and let $S$ and $S'$ be the maximal addable set of $G$ and $G'$, respectively.
% {\sout{If $G'$ does not have an addable set, we consider it satisfies Case (b), i.e., the size of the maximal addable set strictly decreases (to 0).}}
In the following we will show that $S' \subset S$.

% We note that $G'$ contains all the edges in $G$ and the newly generated equality edges due to Case (4) except for the disappeared equality edges in Case (2). Next we will show that the maximal addable set $S$ in $G$ and the maximal addable set $S'$ in $G'$ (if exists) satisfying $S' \subset S$ which corresponds to (b) in the lemma statement and completes the proof.

We first show $S' \subseteq S$.
Suppose otherwise, we will show that $S \cup S'$ is also an addable set in $G$, which contradicts to the maximality of $S$.
The reasons that $S \cup S'$ is an addable set in $G$ are as follows.
\begin{enumerate}[label=(\roman*)]
\item %Because $G$ and $G'$ share the same set of envy edges, we know there will be no envy edges pointing to either $S$ or $S'$ in $G$.
We have already proved that {compared} to $G$, there is no new envy edge in $G'$. If $G$ and $G'$ has the same number of envy edges, they must share exactly the same set of envy edges.
Hence, there will be no envy edge pointing to either $S$ or $S'$ in $G$.

\item %If there is an equality edge pointing from $N \setminus \{S \cup S'\}$ to $S \cup S'$ in $G$, this edge will still be present in $G'$. Because $S$ and $S'$ are addable sets in $G$ and $G'$ respectively, we know there will be no equality edges from $N \setminus \{S \cup S'\}$ to $S \cup S'$ neither.
If there is an equality edge from $N \setminus (S \cup S')$ to $S \cup S'$ in $G$, this equality edge cannot be from $N \setminus (S \cup S')$ to $S$ because $S$ is an addable set in $G$.
Hence, it must be from $N \setminus (S \cup S')$ to $S' \setminus S$.
This equality edge remains in $G'$ because neither the agents in $N \setminus (S \cup S')$ nor the agents in $S' \setminus S$ receive any good.
However, this is impossible because $S'$ is an addable set in $G'$.
In summary, there cannot be any equality edges from $N \setminus (S \cup S')$ to $S \cup S'$ in $G$.
%Clearly, there is no edge from $N \setminus (S \cup S')$ to $S$ in $G$. We consider $S'$. Suppose that there exist some edges from $N \setminus (S \cup S')$ to $S'$ in $G$. Since in $G'$ the disappeared edges belong to Case (2) which always point from $S$, these edges will still present in $G'$ which contradicts $S'$ is an addable set in $G'$. Thus there is also no edge from $N \setminus (S \cup S')$ to $S'$ in $G$.
\end{enumerate}

To further prove $S' \subset S$, we recall that according to our algorithm, at least one equality edge, from agent $i^*$ in $N \setminus S$ to some agent $j \in S$, will be included in $G'$.
It is then clear that $j$ cannot be in $S'$.
This concludes the proof.
\end{proof}

\begin{lemma}\label{lem:envyEdgeDecrease}
After the algorithm completes an envy-cycle-elimination phase, the number of envy edges strictly decreases.
\end{lemma}
\begin{proof}
The basic idea of this proof follows from~\citet{LiptonMaMo04}, albeit only strict envy edges were considered in their context.
In the envy-cycle-elimination phase, an envy cycle $T$ is eliminated by giving agent $j$'s bundle to agent $i$ for each edge $i \envyArrow j$ or $i \eqArrow j$ in the cycle.
First, this process does not affect the bundles of agents in $N \setminus T$, hence the set of envy edges among them remains the same.
Next, since we only swap bundles in this phase, the number of envy edges from $N \setminus T$ to $T$ remains the same.
In addition, every agent $i \in T$ receives a weakly better bundle, meaning that the number of envy edges from $T$ to $N \setminus T$ does not increase.
%Finally, because $T$ contains at least one envy edge, that envy edge will be eliminated after the phase.
{Finally, because $T$ contains at least one envy edge, some agent in $T$ will receive a strictly better bundle.
As a result, although some envy edges between agents in $T$ may still exist, the total number of envy edges will decrease by at least one.}
\end{proof}

\begin{lemma}\label{lem:EFhalfTermination}
Algorithm~\ref{alg:EFhalf} terminates in polynomial time with $O(n^3)$ calls to the \perfect oracle and $O(n^4)$ Robertson-Webb queries.
% \sout{Algorithm~\ref{alg:EFhalf} terminates in polynomial time with $O(n^3)$ calls to the perfect allocation oracle.}
\end{lemma}
\begin{proof}
\textbf{Calls to the} \perfect \textbf{oracle.}
By Invariant~\ref{item:addableSetOrEnvyCycle}, each round in Algorithm~\ref{alg:EFhalf} executes either a cake-adding phase or an envy-cycle-elimination phase.
According to Lemmas~\ref{lem:equalEdgeIncrease} and~\ref{lem:envyEdgeDecrease}, the number of envy edges never increases.
Thus the number of rounds in which the number of envy edges strictly decreases is bounded by $O(n^2)$.

We now upper bound the number of cake-adding phase rounds between any two consecutive rounds that decrease the number of envy edges.
If the whole remaining cake is allocated (Step~\ref{EFhalfALG-EF}), $\perfect(C, n, N)$ is called once and then Algorithm~\ref{alg:EFhalf} terminates.
In the case that a piece of remaining cake is allocated, by Lemma~\ref{lem:equalEdgeIncrease}, the size of the maximal addable set strictly decreases
or an addable set no longer exists; in the latter case, the algorithm proceeds to an envy-cycle-elimination phase.
Because the size of any addable set is $O(n)$, it means that the number of cake-adding phase rounds between any two consecutive rounds that decrease the number of envy edges is $O(n)$.

Finally, it follows that Algorithm~\ref{alg:EFhalf} executes at most $O(n^2) \cdot O(n) = O(n^3)$ cake-adding phase rounds.
Every such round calls the \perfect oracle once.
Algorithm~\ref{alg:EFhalf} makes $O(n^3)$ calls to the \perfect oracle.

\medskip

\noindent\textbf{Polynomial running time and RW queries.}
Note that during the algorithm's run, we add resources to a bundle and rotate bundles among agents, but never split a bundle.
For example, the partition of indivisible goods is computed in Step~\ref{EFhalfALG-EF1Allocation} and remains the same since then.
To avoid redundant computations, we maintain an $n$ by $n$ array to keep track of $u_i(A_j)$ for all $i, j \in N$ and update them as necessary.

In Step~\ref{EFhalfALG-EF1Allocation}, finding an EF1 allocation of indivisible goods can be done in $O(m n \log m)$ via the round-robin algorithm~\citep{CaragiannisKuMo19}.
The implementation details are as follows.
We first compute the sorted order of goods according to each agent's valuation, which takes $O(n m\log m)$ time overall.
Next, in each agent's turn, we keep looking for the next unallocated good in that agent's sorted list.
This step takes $O(mn)$ time in total.
Therefore, the overall running time of the round-robin algorithm is dominated by $O(m n \log m)$.
% We immediately update the values of all agents for the bundle of the agent who are currently picking her most preferred indivisible good; it takes $O(mn)$ time, which is dominated by $O(m n \log m)$.

Next, in Step~\ref{EFhalfALG-initialEnvyGraph}, the overall time to construct the corresponding envy graph is $O(n^2)$.

We now consider the \texttt{while} loop.
According to Lemma~\ref{lem:max_addable}, we can find the maximal addable set or decide its non-existence in time $O(n^3)$.
%In the case that we need to perform an envy-cycle-elimination, an envy cycle $T$ can be found by tracking, via depth-first search (DFS), whether agent $i$ can be reached from $j$ for each envy edge $i \envyArrow j$.
In the case that we need to perform an envy-cycle-elimination, an envy cycle $T$ can be found in the following way.
Fix an agent $i$, we can first spend $O(n)$ time scanning all outgoing edges and ignore those equality edges.
Then, we apply depth-first search (DFS) starting from vertex $i$.
If there is a back edge pointing to vertex $i$, then there must be an envy cycle with at least one envy edge, say, e.g., $i \envyArrow j$, for some $j \in N$.
This takes $O(n^2)$ time since DFS dominates the time complexity.
Since there are $O(n)$ agents, overall, this step can be implemented in $O(n^3)$ time.

In the following, we discuss the steps in each phase at length.
\begin{description}
\item[Cake-adding phase] When we have $S = N$ satisfied in Step~\ref{EFhalfALG-EFbegin}, we implement an EF allocation by calling $\perfect(C, n, N)$.
It takes $O(n)$ time to update the allocation.
Algorithm~\ref{alg:EFhalf} then terminates.

It takes $O(n^2)$ time in Step~\ref{EFhalfALG-deltai} to compute $\delta_i$ for all $i \in N \setminus S$.
Steps~\ref{EFhalfALG-chooseCakeBegin} and~\ref{EFhalfALG-xi} need $O(n)$ evaluation and cut queries respectively.
Once $C'$ is determined in Step~\ref{EFhalfALG-cakeEnd}, we can make $O(n)$ evaluation queries from all $n$ agents over $C'$.
Because we use a perfect allocation of $C'$, we can directly compute $u_i(C') / |S|$ for all $i \in N$ to obtain the value increment of each agent in the addable set.
It then takes $O(n^2)$ time to update all agents' valuations of all bundles after Step~\ref{EFhalfALG-PerfectUpdate}.
After this, updating an envy graph also takes $O(n^2)$ time.

Since we only make RW queries in this phase, we summarize here that Algorithm~\ref{alg:EFhalf} makes $O(n^4)$ Robertson-Webb queries, because there are $O(n^3)$ cake-adding phases (stated earlier in this proof) and each such phase needs at most $O(n)$ RW queries.

\item[Envy-cycle-elimination phase] Since we maintain an array as the reference for agents' valuations over the current bundles, we can rotate the bundles as well as update the array and the envy graph in time $O(n^2)$.
\end{description}
The {remaining} steps can be implemented in time $O(n)$.
Overall, Algorithm~\ref{alg:EFhalf} runs in time $O(m n \log m + n^6)$, where the $n^6$ term comes from the $O(n^3)$ total number of  \texttt{while} loops and $O(n^3)$ time to run each loop.
\end{proof}

Finally the correctness of Theorem~\ref{thm:existenceOfEFhalf} is directly implied by Lemma~\ref{lem:EFhalfAllocation} and Lemma~\ref{lem:EFhalfTermination}.

\paragraph{Bounded Protocol in the RW Model}
Even though we showed that Algorithm~\ref{alg:EFhalf} can produce an EFM allocation, it is not a bounded protocol in the RW model.
This is because our algorithm utilizes an oracle that can compute a perfect allocation of any piece of cake.
However, while {a perfect allocation} always exists, it is known that {such an} allocation cannot be implemented with a finite number of queries in the RW model, even if there are only two agents~\citep{RobertsonWe98}.
Whether there exists a \emph{bounded} protocol in the RW model to compute an EFM allocation remains a very interesting open question.
Note that the perfect allocation oracle cannot be implemented even with a finite number of queries, therefore it is even an open question to find a \emph{finite} EFM protocol.

A natural and tempting approach to get a bounded EFM protocol would be to use an envy-free allocation, for which {a bounded protocol in the RW model} is known~\citep{AzizMa16}, to replace the perfect allocation in Step~\ref{EFhalfALG-addCake}.
Note that doing so would not create any new envy edge within set $S$.
Then, in order to not create any envy edge from $N \setminus S$ to $S$, we need to restrict the total value of the piece of cake allocated to $S$ to not exceed $\delta_i$ (Step~\ref{EFhalfALG-deltai}) for every agent $i \in N \setminus S$.
% A partial EFM allocation remains to be done by carefully choosing the piece of cake so that after cake-adding, an envy-edge would never be introduced from $N \setminus S$ to $S$ either.
% An important observation is that to compute an EF allocation with respect to $S$, only the valuations of agents in $S$ are taken into account.
% It is possible that some piece allocated to an agent in $S$ might be significantly valuable to an agent in $N \setminus S$ and the latter agent would envy the former after cake-adding.
% In order to circumvent this issue, we could restrict the total value of the piece of cake allocated to $S$ to not exceed $\delta_i$ (Step~\ref{EFhalfALG-deltai}).
However, when doing so, we will not be able to quantify the progress of the algorithm like in Lemma~\ref{lem:equalEdgeIncrease}.
Specifically, we can no longer guarantee that either the number of envy edges strictly decreases or the size of the maximal addable set strictly decreases.
This is because we are not guaranteed the equality edge from agent $i^*$ to some agent $j \in S$ as we rely on in the proof of Lemma~\ref{lem:equalEdgeIncrease}.
In other words, we cannot show the algorithm will always terminate in bounded steps.
Interestingly, in sharp contrast, in Section~\ref{sec:approxEFM} we will show that an approximate envy-free protocol, instead of an $\epsilon$-perfect protocol, is enough to give an efficient $\epsilon$-EFM algorithm.
We will discuss this phenomenon in further detail in Section~\ref{sec:approxEFM}.

In the next two sections, we present two bounded protocols to compute an EFM allocation for two special cases, and another bounded (polynomial time) protocol to compute an $\epsilon$-EFM allocation in the general case.

\section{EFM Allocation in Special Case}\label{sec:EFMSpecial}
In this section, we show two special cases where an EFM allocation can be computed in polynomial time without using the perfect allocation oracle.
One is the 2-agent case with general valuations while the other deals with the $n$-agent case but each agent has a structured density function for the cake.

\subsection{Two Agents}
%We first show that with only two agents, an EFM allocation can be found using a simple cut-and-choose type of algorithm: we begin with an EF1 allocation $(M_1, M_2)$ of all indivisible goods. Assume without loss of generality that $u_1(M_1) \geq u_1(M_2)$. Next agent 1 adds the cake into $M_1$ and $M_2$ so that the two bundles are as close to each other as possible. Note that if $u_1(M_1) > u_1(M_2 \cup C)$, agent 1 would add all cake to $M_2$. If $u_1(M_1) \leq u_1(M_2 \cup C)$, agent 1 has a way to make the two bundles equal.
We first show that with only two agents, an EFM allocation can be found using a simple cut-and-choose type of algorithm.
We start with a partition $(M_1, M_2)$ of all indivisible goods such that agent 1 is EF1 with respect to either bundle.
{Without loss of generality,} we assume that $u_1(M_1) \geq u_1(M_2)$.
Next agent 1 adds the cake into $M_1$ and $M_2$ so that the two bundles are as close to each other as possible.
Note that if $u_1(M_1) > u_1(M_2 \cup C)$, agent 1 would add all cake to $M_2$.
If $u_1(M_1) \leq u_1(M_2 \cup C)$, agent 1 has a way to make the two bundles equal.
We then give agent 2 her preferred bundle and leave to agent 1 the remaining bundle.

\begin{algorithm}[t]
\caption{EFM Allocation for Two Agents}
\label{alg:2Agents}
\begin{algorithmic}[1]
\REQUIRE Agents 1 and 2, indivisible goods $M$ and cake $C$.

\STATE Divide $M$ into two parts, $M_1, M_2$, such that agent 1 is EF1 with respect to either bundle. Assume w.l.o.g.~that $u_1(M_1) \geq u_1(M_2)$ (otherwise we can swap $M_1$ and $M_2$). \label{2AgentALG-EF1}
\IF {$u_1(M_1) \leq u_1(M_2 \cup C)$} \label{2AgentALG-cakeNotEnoughBegin}
	\STATE Let agent 1 partition the cake into two pieces $(C_1, C_2)$, such that $u_1(M_1 \cup C_1) = u_1(M_2 \cup C_2)$. \label{2AgentALG-divideCake}
	\STATE Let $(A_1, A_2) = (M_1 \cup C_1, M_2 \cup C_2)$.
\ELSE \label{2AgentALG-cakeEnough}
	\STATE Let $(A_1, A_2) = (M_1, M_2 \cup C)$.
\ENDIF

\STATE Give agent 2 her preferred bundle among $A_1, A_2$. Give agent 1 the remaining bundle. \label{2AgentALG-allocation}
\end{algorithmic}
\end{algorithm}

\begin{theorem}
Algorithm \ref{alg:2Agents} returns an EFM allocation in the case of two agents in polynomial time.
\end{theorem}
\begin{proof}
\noindent\textbf{Correctness.}
It is obvious that all goods are allocated.
We next show that the allocation returned is EFM.
Agent 2 is guaranteed EF (thus EFM) since she gets her preferred bundle between $A_1$ and $A_2$.
In the following, we focus on agent 1.
If $u_1(M_1) \leq u_1(M_2 \cup C)$ holds, agent 1 is indifferent between bundles $A_1$ and $A_2$, so either $A_1$ or $A_2$ makes her EF (thus EFM).
In the case that $u_1(M_1) > u_1(M_2 \cup C)$ holds, agent 1 is EF if she receives $A_1$, and is EFM if she gets $A_2$ because $A_1$ consists of only indivisible goods and there exists some good $g$ in $A_1$ such that $u_1(A_2) \geq u_1(M_2) \geq u_1(A_1 \setminus \{g\})$.

\medskip

\noindent\textbf{Polynomial time.}
To obtain the initial partition $(M_1, M_2)$, we can let two copies of agent 1 run the round-robin algorithm on the indivisible items.
% The resulting bundles are EF1 from the perspective of agent 1.
This step can be done easily in polynomial time.
% Then, Step~\ref{2AgentALG-divideCake} requires only one cut query, which is to let agent 1 return a point $x$ such that $u_1([0, x]) = 1/2 - u_1(M_1)$
% (recall that we assume normalized utilities).
% Step~\ref{2AgentALG-allocation} requires one evaluation query; we can without loss of generality let agent 2 return $u_2([a, b])$, where $[a, b]$ is the piece of cake in bundle $A_1$, and then calculate the utilities of bundle for agent 2.}
The {remaining} steps only take constant running time and a constant number of RW queries.
The conclusion follows.
\end{proof}

\paragraph{A Stronger EFM Notion}
With two agents, an \emph{envy-freeness up to any good (EFX)} allocation, in which no agent prefers the bundle of another agent following the removal of \emph{any} single good, always exists~\citep{PlautRo20}.
This result can be carried over to show the existence of a stronger EFM notion in the mixed goods setting, in which an agent is EFX towards any agent with only indivisible goods, and EF towards the rest.
Such an allocation can be obtained by using an EFX partition (with respect to agent 1) instead of an EF1 partition in Step~\ref{2AgentALG-EF1} of Algorithm~\ref{alg:2Agents}.
Moreover, with any number of agents, whenever an EFX allocation exists among indivisible goods,\footnote{EFX exists for three agents~\citep{ChaudhuryGaMe20} or $n$ agents with identical valuations~\citep{PlautRo20}, but the existence of EFX remains open for four or more agents with additive valuations.}
we can start with such an EFX allocation in Step~\ref{EFhalfALG-EF1Allocation} of Algorithm~\ref{alg:EFhalf}.
The cake-adding phase maintains the EFM condition and does not introduce new envy.
Thus Algorithm~\ref{alg:EFhalf} will also produce an allocation with this stronger notion of EFM.

\subsection{Any Number of Agents with Piecewise Linear Functions}
In the second case, we consider an arbitrary number of agents when agents' valuation functions over the cake are \emph{piecewise linear}.

\begin{definition}
A valuation density function $f_i$ is \emph{piecewise linear} if the interval $[0, 1]$ can be partitioned into a finite number of intervals such that $f_i$ is linear on each interval.
\end{definition}

Piecewise linear function is a generalization of both \emph{piecewise uniform function} and \emph{piecewise constant function}, each of which has been considered in several previous fair division works~\citep{BeiChHu12,ChenLaPa13,BeiHuSu20}.
In this case, we do not use the RW model, but rather assume that the valuation functions are provided to us in \emph{full information}.

The only obstacle in converting Algorithm~\ref{alg:EFhalf} to a bounded protocol is the implementation of the perfect allocation oracle for cake cutting.
{When} agents have piecewise linear functions, \citet{ChenLaPa13} showed that a perfect allocation can be computed efficiently in polynomial time.
This fact, combined with Theorem~\ref{thm:existenceOfEFhalf}, directly implies the following result.

\begin{corollary}
For any number of agents with piecewise linear density functions over the cake, an EFM allocation can be computed in polynomial time.
\end{corollary}

\section{$\epsilon$-EFM: Algorithm}\label{sec:approxEFM}
In this section, we focus on $\epsilon$-EFM, a relaxation of the EFM condition.
Despite the computational issues with finding bounded exact EFM protocols, we will show that there is an efficient algorithm in the RW model that computes an $\epsilon$-EFM allocation for general density functions with running time polynomial in $n$, $m$ and $1/\epsilon$.

Since the difficulty in finding a bounded EFM protocol in the RW model lies in computing perfect allocations of a cake (Section~\ref{sec:EFMExistence}), one might be tempted to simply use a bounded \emph{$\epsilon$-Perfect Allocation} protocol to replace the exact procedure. Here a partition $\mathcal{C} = (C_1, C_2, \dots, C_k)$ of cake $C$ is said to be \emph{$\epsilon$-perfect} if for all $i \in N$, $j \in [k]$, $|u_i(C_j) - u_i(C)/k| \leq \epsilon$.
However, although a bounded $\epsilon$-perfect protocol exists in the RW model~\citep{RobertsonWe98,BranzeiNi17}, all known protocols {have} running time exponential in $1/\epsilon$.\footnote{\citet{BranzeiNi17} showed that an $\epsilon$-perfect allocation can be computed in $O(n^3/\epsilon)$ RW queries. However, although the query complexity is polynomial, the protocol still requires an exponential \emph{running time} because it finds the correct partition of the small pieces into bundles via an exhaustive enumeration, of which no polynomial time algorithm is known.}
% We elaborate this point further.
% The protocols usually start with cutting pieces that all agents agree are small, which can be done with RW queries polynomial in $1/\epsilon$.
% Next, the protocols partition these small pieces into bundles such that the resulting allocation is $\epsilon$-perfect by leveraging \emph{existence} results which, to the best of our knowledge, can be done by exhaustive enumeration but no known efficient algorithmic result exists.}\footnote{Specifically, the protocol designed by \citet{RobertsonWe98} needs to arrange $t$ small pieces in order, say $X_1, X_2, \dots, X_t$, such that for each $k \in [t]$ the values $u_i(\bigcup_{j=1}^{k} X_j)$ are nearly equal for all $i \in N$.
% Then, it would be able to partition the list of $t$ pieces into desired sublists.
% We further note that \citeauthor{RobertsonWe98}'s protocol works for a more general fairness concept called \emph{$\epsilon$-exact division}.
% We now turn our attention to the protocol~\citep{BranzeiMi15} pointed by \citet{BranzeiNi17}.
% \citet{BranzeiMi15}'s protocol enumerates all subsets with at most $n(k-1)$ cut points from all cuts made by agents when they discretized the cake; it then check each allocation defined by the selective cuts, i.e., if the allocation is desired.
% This protocol computes an $\epsilon$-perfect allocation in that the cuts made by agents are close to the cut points of a perfect allocation on the original cake and \citet{Alon87} showed that $n(k-1)$ cuts are enough to produce $k$ pieces that are equal for all $n$ agents.}
It is still an open question to find an $\epsilon$-perfect allocation with both query and time complexity polynomial in $1/\epsilon$.
Therefore, to design an efficient $\epsilon$-EFM protocol, extra work needs to be done to circumvent this issue.

We next define the relaxed version of EF and envy graph.

\begin{definition}[$\epsilon$-EF]
An allocation $\mathcal{A}$ is said to satisfy \emph{$\epsilon$-envy-freeness ($\epsilon$-EF)} if for all agents $i, j \in N$, $u_i(A_i) \geq u_i(A_j) - \epsilon$.
\end{definition}

\begin{definition}[$\epsilon$-envy graph]\label{def:epsEnvyGraph}
Given an allocation $\mathcal{A}$ and a parameter $\epsilon$, the $\epsilon$-envy graph is defined as $G(\epsilon) = (N, E_{\epsilon\text{-envy}} \cup E_{\epsilon\text{-eq}})$, where every vertex represents an agent, and $E_{\epsilon\text{-envy}}$ and $E_{\epsilon\text{-eq}}$ consist of the following two types of edges, respectively:
\begin{itemize}
\item $\epsilon$-envy edge: $i \epsEnvyArrow j$ if $u_i(A_i) < u_i(A_j) - \epsilon$;
\item $\epsilon$-equality edge: $i \epsEqArrow j$ if $u_i(A_j) - \epsilon \leq u_i(A_i) \leq u_i(A_j)$.
\end{itemize}
\end{definition}

Given an $\epsilon$-envy graph, a cycle is said to be an \emph{$\epsilon$-envy cycle} if it contains at least one $\epsilon$-envy edge.
We also note that when $\epsilon = 0$, the $\epsilon$-envy graph degenerates into the envy graph defined in Section \ref{sec:EFMExistence}.

\subsection{The Algorithm}
The complete algorithm to compute an $\epsilon$-EFM allocation is shown in Algorithm~\ref{alg:epsEFhalf}.
Similarly to Algorithm~\ref{alg:EFhalf}, our algorithm adds resources to the partial allocation iteratively.
We always maintain the partial allocation to be $\hat{\epsilon}$-EFM where $\hat{\epsilon}$ is updated increasingly and would never exceed $\epsilon$.
This will ensure that the final allocation is $\epsilon$-EFM.

\begin{algorithm}[ht!]
\caption{$\epsilon$-EFM Algorithm}
\label{alg:epsEFhalf}
\begin{algorithmic}[1]
\REQUIRE Agents $N$, indivisible goods $M$, cake $C$, and parameter $\epsilon$.

\STATE $\hat{\epsilon} \leftarrow \frac{\epsilon}{4}$, {$\eferror \leftarrow \frac{\epsilon^2}{8n}$} \label{epsEFhalfALG-initialization}
\STATE Find an arbitrary EF1 allocation $(A_1, A_2, \dots, A_n)$ of $M$ to $n$ agents. \label{epsEFhalfALG-EF1Allocation}
\STATE Construct an $\hat{\epsilon}$-envy graph $G(\hat{\epsilon}) = (N, E_{\hat{\epsilon}\text{-envy}} \cup E_{\hat{\epsilon}\text{-eq}})$ accordingly. \label{epsEFhalfALG-initialEpsEnvyGraph}

\WHILE {$C \neq \emptyset$} \label{epsEFhalfALG-remainingCake}
	\IF {there exists an addable set $S$} \label{epsEFhalfALG-addCakeBegin}
		\STATE // cake-adding phase
		\IF {$S = N$} \label{epsEFhalfALG-S=N}
			\STATE Let $(C_1, C_2, \dots, C_n) = \epsEF{\frac{\epsilon}{4}}(C, N)$. \label{epsEFhalf-epsEF}
			\STATE $C \leftarrow \emptyset$
			\STATE $\hat{\epsilon} \leftarrow \hat{\epsilon} + \epsilon/4$ \label{epsEFhalfALG-hatEpsS=N}
			\STATE Add $C_i$ to bundle $A_i$ for all $i \in N$. \label{epsEFhalfALG-EpsUpdate}
		\ELSE \label{epsEFhalfALG-SneqNBegin}
			\IF {$\max_{i \in N \setminus S} u_i(C) \leq \hat{\epsilon}$} \label{epsEFM-maxValForC}
				\STATE $C' \leftarrow C, C \leftarrow \emptyset$
			\ELSE
				\STATE Suppose w.l.o.g.~that $C = [a, b]$. For each agent $i \in N \setminus S$, if $u_i([a, b]) \geq \hat{\epsilon}$, let $x_i$ be a point such that $u_i([a, x_i]) = \hat{\epsilon}$; otherwise, let $x_i = b$. \label{epsEFhalfALG-cakeBegin}
				\STATE $i^* \leftarrow \arg\min_{i \in N \setminus S} x_i$
				\STATE $C' \leftarrow [a, x_{i^*}], C \leftarrow C \setminus C'$  \label{epsEFhalfALG-cakeEnd}
			\ENDIF
			\STATE {Let $(C_1, C_2, \dots, C_k) = \epsEF{\eferror}(C', S)$ where $k = |S|$.} \label{epsEFhalfALG-addCake}
			\STATE {$\hat{\epsilon} \leftarrow \hat{\epsilon} + \eferror$} \label{epsEFhalfALG-increaseEps}
			\STATE Add $C_i$ to the bundle of the $i$-th agent in $S$. \label{epsEFhalfALG-EpsPrimeUpdate}
			\STATE Update $\hat{\epsilon}$-envy graph $G(\hat{\epsilon})$ accordingly. \label{epsEFhalfALG-SneqNEnd}
		\ENDIF \label{epsEFhalfALG-addCakeEnd}
	\ELSE \label{epsEFhalfALG-envyCycleBegin}
		\STATE // envy-cycle-elimination phase
		\STATE Let $T$ be an $\hat{\epsilon}$-envy cycle in the $\hat{\epsilon}$-envy graph.
		\STATE For each agent $j \in T$, give agent $j$'s whole bundle to agent $i$ who points to her in $T$. \label{epsEFhalfALG-cycleElimination}
		\STATE Update $\hat{\epsilon}$-envy graph $G(\hat{\epsilon})$ accordingly.
	\ENDIF \label{epsEFhalfALG-envyCycleEnd}
\ENDWHILE \label{epsEFhalfALG-remainingEnd}

\RETURN $(A_1, A_2, \dots, A_n)$ \label{epsEFhalfALG-lastReturn}
\end{algorithmic}
\end{algorithm}

Like Algorithm~\ref{alg:EFhalf}, starting with an EF1 allocation of indivisible goods to all agents in Step~\ref{epsEFhalfALG-EF1Allocation}, Algorithm~\ref{alg:epsEFhalf} then executes in rounds (Steps \ref{epsEFhalfALG-remainingCake}-\ref{epsEFhalfALG-remainingEnd}).
Even though each round still executes either a cake-adding phase or an envy-cycle-elimination phase, the execution details are different from Algorithm~\ref{alg:EFhalf}.

\begin{itemize}
\item In the cake-adding phase, instead of allocating some cake to an addable set $S$ in a way that is perfect, we resort to a {$\eferror$-EF} allocation, where {$\eferror$} will be fixed later in Algorithm~\ref{alg:epsEFhalf}.
In the following, we will utilize an algorithm {$\epsEF{\eferror}(C, S)$} that could return us a {$\eferror$-EF} allocation for any set of agents $S$ and cake $C$.
Note that, for any $\bar{\epsilon} > 0$, the algorithm $\epsEF{\bar{\epsilon}}$ can be implemented with both running time and query complexity polynomial in the number of agents involved and $1/\bar{\epsilon}$~\citep{Procaccia16}.
We also update $\hat{\epsilon}$ to a larger number, say {$\hat{\epsilon} + \eferror$}, in order to avoid generating $\hat{\epsilon}$-envy edges due to cake-adding.

\item In the envy-cycle-elimination phase, we eliminate an $\hat{\epsilon}$-envy cycle, instead of an envy cycle, by rearranging the current partial allocation.
\end{itemize}

% We remark that when there exist only indivisible goods, Algorithm~\ref{alg:epsEFhalf} performs Steps \ref{epsEFhalfALG-initialization}-\ref{epsEFhalfALG-initialEpsEnvyGraph} and terminates with an EF1 allocation (which is also EFM). When there exist only divisible goods, in Step~\ref{epsEFhalf-epsEF} Algorithm~\ref{alg:epsEFhalf} distributes the whole cake and ends with an $\epsilon/4$-EF allocation of the cake (which is again $\epsilon$-EFM).

\subsection{Analysis}
Our main result for the $\epsilon$-EFM allocation is as follows:

\begin{theorem}\label{thm:epsEFM}
% \sout{An $\epsilon$-EFM allocation can be found by Algorithm~\ref{alg:epsEFhalf} in time $O(n^5 / \epsilon^5 + m^2)$.}
An $\epsilon$-EFM allocation can be found by Algorithm~\ref{alg:epsEFhalf} with running time $O(n^4/\epsilon + m n \log m)$, $O(n^3/\epsilon)$ Robertson-Webb queries,
% at most one call to the $\epsEF{\bar{\epsilon}}$ oracle with $\bar{\epsilon} = \frac{\epsilon}{4}$,
and $O(n/\epsilon)$ calls to the approximate $\textsc{EFAlloc}$ oracle.
\end{theorem}

To prove Theorem~\ref{thm:epsEFM}, we first show that each round (Steps~\ref{epsEFhalfALG-remainingCake}-\ref{epsEFhalfALG-remainingEnd} in Algorithm~\ref{alg:epsEFhalf}) maintains the following invariants during the run of the algorithm.

\vbox{
\hrule
\paragraph{Invariants}
\begin{enumerate}[label=\textbf{B\arabic*.},ref=B\arabic*]
\item In each round there is either an addable set for the cake-adding phase or an $\hat{\epsilon}$-envy cycle for the envy-cycle-elimination phase. \label{item:addableSetOrepsEnvyCycle}
% That is, if there exists any addable set, we can find an addable set easily; otherwise, there must exist an $\hat{\epsilon}$-envy cycle.
\item The partial allocation is always $\hat{\epsilon}$-EFM with the current $\hat{\epsilon}$. \label{item:partial-epsHatEFM}
\end{enumerate}
\hrule
}

We next prove these invariants in the following.

\begin{lemma}
Invariant~\ref{item:addableSetOrepsEnvyCycle} holds during the algorithm's run.
\end{lemma}
\begin{proof}
The proof is similar to the proof of Lemma \ref{lem:envy-cycle-addable-set}, except that we consider the $\hat{\epsilon}$-envy edge instead of the envy edge.
\end{proof}

\begin{lemma}\label{lem:alwaysepsEFhalf}
Invariant~\ref{item:partial-epsHatEFM} holds during the algorithm's run.
\end{lemma}
\begin{proof}
First, it is worth noting that at the beginning, when indivisible goods are allocated, the allocation is EF1 and therefore EFM.
We then note that the partial allocation is only updated in Steps~\ref{epsEFhalfALG-EpsUpdate} and~\ref{epsEFhalfALG-EpsPrimeUpdate} in the cake-adding phase as well as in Step~\ref{epsEFhalfALG-cycleElimination} in the envy-cycle-elimination phase.
Given a partial allocation that is $\hat{\epsilon}$-EFM, we will show that each of these updates maintains $\hat{\epsilon}$-EFM with the updated $\hat{\epsilon}$, which completes the proof of Lemma~\ref{lem:alwaysepsEFhalf}.
We note that $\hat{\epsilon}$ is only updated in the cake-adding phase and is non-decreasing during the algorithm's run.

For analysis in the envy-cycle-elimination phase, the proof is identical to that of Lemma~\ref{lem:alwaysEFhalf} in the case of envy-cycle-elimination phase.

We then discuss the cases in the cake-adding phase.
For the updated partial allocation in Step~\ref{epsEFhalfALG-EpsUpdate}, we allocate the remaining cake $C$ in a way that is $\frac{\epsilon}{4}$-EF which implies that $u_i(C_i) \geq u_i(C_j) - \epsilon / 4$ holds for any pair of agents $i, j \in N$.
Given the partial allocation that is $\hat{\epsilon}$-EFM, we have
\begin{equation*}
u_i(A_i \cup C_i) \geq u_i(A_j \cup C_j) - \hat{\epsilon} -\epsilon/4.
\end{equation*}
Thus it is clear that no $\hat{\epsilon}$-envy edge will be generated among agents in $S$ if we update $\hat{\epsilon}$ to $\hat{\epsilon} + \epsilon / 4$ in Step~\ref{epsEFhalfALG-hatEpsS=N}.

For the updated partial allocation in Step~\ref{epsEFhalfALG-EpsPrimeUpdate}, we allocate some cake $C'$ to agents in $S$ in a way that is {$\eferror$-EF}.
By a similar argument to the case above, we know that no $\hat{\epsilon}$-envy edge will be generated among agents in $S$ if we update $\hat{\epsilon}$ to {$\hat{\epsilon} + \eferror$}.
Then, for any agent $i \in N \setminus S$, we have $u_i(A_i) > u_i(A_j)$ where $j \in S$.
As $u_i(C') \leq \hat{\epsilon}$, we have
\begin{equation*}
u_i(A_i) > u_i(A_j \cup C'_j) - \hat{\epsilon},
\end{equation*}
where $C'_j$ is the piece of cake allocated to agent $j \in S$.
It means that again no $\hat{\epsilon}$-envy edge will be generated from $N \setminus S$ to $S$ when we update $\hat{\epsilon}$ to {$\hat{\epsilon} + \eferror$}.
Last, it is obvious that for any pair of agents in $N \setminus S$, we do not generate any {$(\hat{\epsilon} + \eferror)$-envy edge} because the bundles of these agents remain the same.
We note that $\hat{\epsilon}$ is updated to {$\hat{\epsilon} + \eferror$} in Step~\ref{epsEFhalfALG-increaseEps} in order to make sure that there is no introduced $\hat{\epsilon}$-envy edge in the updated $\hat{\epsilon}$-envy graph.
\end{proof}

\paragraph{Correctness}
\begin{lemma}\label{lem:valSumDecrease}
In the cake-adding phase, if the piece of cake to be allocated is not the whole remaining cake, the sum of all agents' valuations on the remaining cake decreases by at least $\hat{\epsilon}$.
\end{lemma}
\begin{proof}
We consider the cake-adding phase in Steps~\ref{epsEFhalfALG-addCakeBegin}-\ref{epsEFhalfALG-addCakeEnd}.
If the piece of cake $C'$ to be allocated is not the whole remaining cake, there exists an agent $i^*$ in $N \setminus S$ such that $u_{i^*}(C') = \hat{\epsilon}$ according to Steps~\ref{epsEFhalfALG-cakeBegin}-\ref{epsEFhalfALG-cakeEnd}.
Thus the lemma follows.
\end{proof}

We are now ready to show the correctness of Algorithm~\ref{alg:epsEFhalf}.

\begin{lemma}\label{lem:epsEFhalfAllocation}
Algorithm~\ref{alg:epsEFhalf} always returns an $\epsilon$-EFM allocation upon termination.
\end{lemma}
\begin{proof}
By Invariant~\ref{item:partial-epsHatEFM}, it suffices to prove that all goods are allocated and $\hat{\epsilon}$ is at most $\epsilon$ when Algorithm~\ref{alg:epsEFhalf} terminates.
All indivisible goods are allocated in Step \ref{epsEFhalfALG-EF1Allocation}.
Then the \texttt{while} loop (Steps~\ref{epsEFhalfALG-remainingCake}-\ref{epsEFhalfALG-remainingEnd}) terminates only when the cake is also fully allocated, as desired.

We now turn our attention to $\hat{\epsilon}$.
First, $\hat{\epsilon}$ is initialized to be $\epsilon/4$ and never decreases during the algorithm's run.
If the whole remaining cake is allocated, there is at most one execution of the cake-adding phase (Steps~\ref{epsEFhalfALG-S=N}-\ref{epsEFhalfALG-EpsUpdate}).
Moreover, $\hat{\epsilon}$ is increased by $\epsilon/4$ in Step~\ref{epsEFhalfALG-hatEpsS=N}.
We will show later in this proof that this increment would not let $\hat{\epsilon}$ exceed $\epsilon$.
We then focus on the case where the remaining cake is not fully allocated.
There are at most $n/\hat{\epsilon} \leq 4n/\epsilon$ executions of the cake-adding phase (Steps~\ref{epsEFhalfALG-SneqNBegin}-\ref{epsEFhalfALG-SneqNEnd}) according to Lemma~\ref{lem:valSumDecrease} and the fact that agents' utilities are normalized to 1.
In addition, $\hat{\epsilon}$ is increased by {$\eferror$} in each cake-adding phase in Step~\ref{epsEFhalfALG-increaseEps}.
{Thus, $\hat{\epsilon}$ is upper bounded by $\epsilon/4 + 4n/\epsilon \cdot \eferror + \epsilon/4 = \epsilon$ due to $\eferror = \frac{\epsilon^2}{8n}$.}
It follows that the final allocation is $\epsilon$-EFM.
\end{proof}

\paragraph{Termination and Time Complexity}
\begin{lemma}\label{lem:socialWelfareIncrease}
In the envy-cycle-elimination phase, the social welfare $\sum_{i \in N} u_i(A_i)$, increases by at least $\hat{\epsilon}$.
\end{lemma}
\begin{proof}
We eliminate an $\hat{\epsilon}$-envy cycle $T$ which contains at least one $\hat{\epsilon}$-envy edge in the envy-cycle-elimination phase (Steps~\ref{epsEFhalfALG-envyCycleBegin}-\ref{epsEFhalfALG-envyCycleEnd}).
Step~\ref{epsEFhalfALG-cycleElimination} eliminates the cycle by giving agent $j$'s bundle to agent $i$ for each edge {$i \xrightarrow{\hat{\epsilon}\text{-ENVY}} j$ or $i \xrightarrow{\hat{\epsilon}\text{-EQ}} j$} in cycle $T$.
None of the agents involved in $T$ is worse off and at least one agent in $T$ is better off by at least $\hat{\epsilon}$ by the definition of $\hat{\epsilon}$-envy edge in Definition \ref{def:epsEnvyGraph}.
Since agents outside cycle $T$ do not change their bundles, we complete the proof.
\end{proof}

\begin{lemma}\label{lem:epsEFhalfTime}
% \sout{Algorithm~\ref{alg:epsEFhalf} runs in $O(n^5/\epsilon^5 + m^2)$.}
Algorithm~\ref{alg:epsEFhalf} has running time $O(n^4/\epsilon + m n \log m)$ and invokes
% at most one call to the $\epsEF{\bar{\epsilon}}$ oracle with $\bar{\epsilon} = \frac{\epsilon}{4}$,
$O(n/\epsilon)$ calls to the approximate $\textsc{EFAlloc}$ oracle and $O(n^3/\epsilon)$ Robertson-Webb queries.
\end{lemma}
\begin{proof}
Since most parts of Algorithm~\ref{alg:epsEFhalf} are similar to those in Algorithm~\ref{alg:EFhalf}, and we have discussed their time complexities in the proof of Lemma~\ref{lem:EFhalfTermination}, we will focus on the steps that affect the time complexity for Algorithm~\ref{alg:epsEFhalf} in this proof.

%The EF1 allocation of indivisible goods can be done in $O(m^2)$ via the round-robin algorithm~\citep{CaragiannisKuMo19}.
Similar to Algorithm~\ref{alg:EFhalf}, Algorithm~\ref{alg:epsEFhalf} takes $O(m n \log m + n^2)$ time to perform Steps~\ref{epsEFhalfALG-EF1Allocation} and~\ref{epsEFhalfALG-initialEpsEnvyGraph}.
Afterwards, by Invariant~\ref{item:addableSetOrepsEnvyCycle}, Algorithm~\ref{alg:epsEFhalf} executes either a cake-adding phase or an envy-cycle-elimination phase in each round, and it takes $O(n^3)$ time to check which phase to go into each time.
% whether there exists an addable set or an $\hat{\epsilon}$-envy cycle in Step \ref{epsEFhalfALG-addCakeBegin}.
% We initialize $\hat{\epsilon}$ to be $\epsilon/4$ in Step~\ref{epsEFhalfALG-initialization} and gradually increase $\hat{\epsilon}$ during the algorithm's run, thus $\hat{\epsilon} \geq \epsilon/4$ in each round.
Next, recall that agents' utilities are normalized to 1 and $\hat{\epsilon}$ is always no less than $\epsilon/4$.
This means there are at most $O(n/\epsilon)$ cake-adding rounds by Lemma~\ref{lem:valSumDecrease} and at most $O(n/\epsilon)$ envy-cycle-elimination rounds by Lemma~\ref{lem:socialWelfareIncrease}.
% This implies that we have $O(n/\epsilon)$ rounds where each takes $O(n^3)$ time to check whether there exists an addable set or an $\hat{\epsilon}$-envy cycle in Step \ref{epsEFhalfALG-addCakeBegin}.

In the following, we discuss the steps in each phase in details.

\begin{description}
\item[Cake-adding phases] When we have $S = N$ in Step~\ref{epsEFhalfALG-S=N}, we invoke $\epsEF{\frac{\epsilon}{4}}$ once, use $O(n)$ time to update the allocation, and terminate the algorithm.
% as stated in Step~\ref{epsEFhalf-epsEF}.
% It then takes $O(n)$ time to update the allocation.
% Algorithm~\ref{alg:epsEFhalf} also terminates.

To determine the piece of cake $C'$ to be allocated later, we need $O(n)$ evaluation queries in Step~\ref{epsEFM-maxValForC}, and $O(n)$ cut queries in Step~\ref{epsEFhalfALG-cakeBegin} if the condition check in Step~\ref{epsEFM-maxValForC} fails.
{In each cake-adding phase, we invoke the $\epsEF{\eferror}$ oracle with $\eferror = \frac{\epsilon^2}{8n}$ once (Step~\ref{epsEFhalfALG-addCake}).}
In order to update agents' valuations for each bundle, we invoke $O(n^2)$ evaluation queries to obtain all agents' valuations of all pieces in $(C_1, C_2, \dots, C_k)$, and then use $O(n^2)$ time to update the envy graph.

Summarizing everything, we conclude that Algorithm~\ref{alg:epsEFhalf} makes $O(n^3/\epsilon)$ Robertson-Webb queries, and $O(n/\epsilon)$ calls to the $\epsEF{\frac{\epsilon^2}{8n}}$ oracle in total in all cake-adding phases.

\item[Envy-cycle-elimination phases] Since we are keeping track of all agents' valuations for all bundles, it takes no Robertson-Webb queries and $O(n^2)$ time to rearrange the bundles as well as update the $\hat{\epsilon}$-envy graph. Overall all envy-cycle-elimination phases take $O(n^3/\epsilon)$ running time.
\end{description}
The {remaining} steps can be implemented in time $O(n)$.
The overall time complexity of our algorithm is $O(n/\epsilon \cdot n^3 + n/\epsilon \cdot n^2 + m n \log m + n^2) = O(n^4/\epsilon + m n \log m)$. The overall Robertson-Webb query complexity is $O(n^3/\epsilon)$, and the number of calls to the {$\epsEF{\eferror}$} oracle is $O(n/\epsilon)$.
\end{proof}

Finally the correctness of Theorem~\ref{thm:epsEFM} is directly implied by Lemmas~\ref{lem:epsEFhalfAllocation} and~\ref{lem:epsEFhalfTime}.

Note that at the end of Section~\ref{sec:EFMExistence}, we explained why {an} exact envy-free oracle may not be helpful to obtain an EFM allocation.
However, as we showed in this section, the approximate envy-free oracle does help to obtain {an} $\epsilon$-EFM allocation.
Lemma~\ref{lem:valSumDecrease} provides the key difference.
In particular, the error allowed in the $\epsilon$-EFM condition {ensures} that agents' welfare for the remaining cake is reduced by at least an amount of $\hat{\epsilon}$.
This claim, however, makes no sense when discussing exact envy-freeness.
Furthermore, Algorithm~\ref{alg:epsEFhalf} introduces additional error into the EFM condition on top of the error that comes from the approximate envy-free oracle.
As a result, even if Algorithm~\ref{alg:epsEFhalf} was paired with an exact envy-free oracle, it would still not produce an exact EFM allocation.

\section{EFM and Efficiency}\label{sec:EFM_PO}
In this section, we discuss how to combine EFM with efficiency considerations.
In particular, we focus on the well-studied efficiency notion of \emph{Pareto optimality}.

\begin{definition}[PO]
An allocation $\mathcal{A}$ is said to satisfy \emph{Pareto optimality (PO)} if there is no allocation $\mathcal{A}'$ that Pareto-dominates $\mathcal{A}$, i.e., {satisfies} $u_i(A'_i) \geq u_i(A_i)$ for all $i \in N$ and at least one inequality is strict.
\end{definition}

\begin{definition}[fPO~\citep{BarmanKrVa18}]
An allocation $\mathcal{A}$ is said to satisfy \emph{fractional Pareto optimality (fPO)} if it is not Pareto dominated by any fractional allocation.\footnote{{In a fractional allocation, an agent may get a fractional share of an \emph{indivisible} good.
We refer to~\citep{BarmanKrVa18} for its formal definition.}}
\end{definition}

As \citet{BarmanKrVa18} noted, an fPO allocation is also PO but not vice versa.
It is known that with divisible resources, an allocation that is both envy-free and PO always exists~\citep{Weller85}.
With indivisible goods, an allocation satisfying both EF1 and fPO (and hence PO) also exists~\citep{BarmanKrVa18}.
Perhaps to our surprise, in the following we show via a counter-example that with mixed types of goods, EFM and PO are no longer compatible.

\begin{example}[EFM is not compatible with PO]\label{ex:EFM-PO-notCompatible}
Consider an instance with two agents, one indivisible good, and one cake.
Agents' valuation functions are listed below.

\begin{center}
\begin{tabular}{|*{3}{c|}}
\hline
& Indivisible good & A cake $C = [0, 1]$ \\
\hline
Agent 1 & $0.6$ & $u_1(C) = 0.4$ with uniform density over $[0, 0.5]$ \\
\hline
Agent 2 & $0.6$ & $u_2(C) = 0.4$ with uniform density over $[0.5, 1]$ \\
\hline
\end{tabular}
\end{center}

It is obvious that in any EFM allocation, one agent will get the indivisible good and the entire cake has to be allocated to the other agent.
However, such an allocation cannot be PO since the agent with the cake has no value for half of it, and giving that half to the other agent would make that agent better off without making the first agent worse off.
\end{example}

This counter-example relies on the fact that in the definition of EFM, if some agent $i$'s bundle contains any positive amount of cake, another agent $j$ will compare her bundle to $i$'s bundle using the stricter EF condition, even if agent $j$ has value zero over $i$'s cake.
This may seem counter-intuitive, because when $j$ has no value over $i$'s cake, removing that cake from $i$'s bundle will not help eliminate agent $j$'s envy.
To this end, one may consider the following weaker version of EFM.

% new (weak) EFM
\begin{definition}[Weak EFM]\label{def:weakEFhalf}
An allocation $\mathcal{A}$ is said to satisfy \emph{weak \longEFM (weak EFM)}, if for any agents $i, j \in N$,
\begin{itemize}
\item if agent $j$'s bundle consists of indivisible goods, and
\begin{itemize}
\item either no divisible good,
\item or divisible goods that yield value 0 to agent $i$, i.e., $u_i(C_j) = 0$,
\end{itemize}
there exists an indivisible good $g \in A_j$ such that $u_i(A_i) \geq u_i(A_j \setminus \{g\})$;
\item otherwise, $u_i(A_i) \geq u_i(A_j)$.
\end{itemize}
\end{definition}

From the definition, it is easy to see that EFM implies weak EFM, which means all existence results for EFM established in Section~\ref{sec:EFMExistence} can be carried over to weak EFM.

This weaker version of EFM precludes the incompatibility result in Example~\ref{ex:EFM-PO-notCompatible}.
Nevertheless, we show in the following example that weak EFM is incompatible with fPO.

\begin{example}[(Weak) EFM is incompatible with fPO]\label{ex:fPO}
Consider an instance with two agents, one indivisible good and two homogeneous divisible goods.
Agents' valuation functions are listed below.

\begin{center}
\begin{tabular}{|*{4}{c|}}
\hline
& Indivisible good & Divisible good 1 & Divisible good 2 \\
\hline
Agent 1 & 2 & 1 & 2 \\
\hline
Agent 2 & 2 & 2 & 1 \\
\hline
\end{tabular}
\end{center}

Because the valuations are symmetric, we can assume without loss of generality that in an EFM allocation, the indivisible good is given to agent 1.
We also observe that in any EFM allocation, we cannot allocate all divisible goods to a single agent.
This means that both agents' bundles must contain some divisible good, which then implies that both agents need to be envy-free towards the other agent's bundle.
Next, via two simple linear programs one can compute the maximum utility of each agent in EFM allocations:
% maximize agent 1's utility and agent 2's utility, respectively, subject to the EF conditions mentioned just now.
giving the indivisible good and one half of divisible good 2 to agent 1 gives her a maximum utility 3; giving divisible good 1 and three quarters of divisible good 2 to agent 2 gives her a maximum utility 2.75.
We note that the maximum utilities for the two agents are achieved under different allocations.
However, even putting these two maximum utilities together, it is dominated by the utilities guaranteed by the fractional allocation in which agent 1 gets divisible good 2 and half of the indivisible good while agent 2 gets divisible good 1 and the other half of the indivisible good, which will give both agents a utility of 3.
This means that any EFM allocation is not fPO in this problem instance.
\end{example}

Could there always exist an allocation that satisfies both weak EFM and PO?
We do not know the answer, and believe this is a very interesting open question.
One tempting approach to answer this open question is to consider the maximum Nash welfare (MNW) allocation.
This is the allocation that maximizes the Nash welfare $\prod_{i \in N} u_i(A_i)$ among all allocations.\footnote{In the case where the maximum Nash welfare is 0, an allocation is an MNW allocation if it gives positive utility to a set of agents of maximal size and moreover maximizes the product of utilities of the agents in that set.}
It has been shown that an MNW allocation enjoys many desirable properties in various settings.
In particular, an MNW allocation is always envy-free and PO in cake cutting~\cite{Segal-HaleviSz19}, and EF1 and PO for indivisible resource allocation~\cite{CaragiannisKuMo19}.
It is therefore a natural idea to conjecture that it also satisfies EFM and PO for mixed goods.
Unfortunately, this is not the case.
Here we give such a counter-example.

\begin{example}[MNW does not imply (weak) EFM]\label{ex:MNWnotImplyEFM}
Consider the following instance with two agents, two indivisible goods and one homogeneous cake.
Agents' valuation functions are listed below.

\begin{center}
\begin{tabular}{|*{4}{c|}}
\hline
& Indivisible good 1 & Indivisible good 2 & A homogeneous cake \\
\hline
Agent 1 & 0.4 & 0.4 & 0.2 \\
\hline
Agent 2 & 0.499 & 0.499 & 0.002 \\
\hline
\end{tabular}
\end{center}

% In an Nash welfare maximizing allocation , both agents must have positive value for their own bundle.
We discuss the following cases to find the {MNW} allocation.
\begin{itemize}
\item When both indivisible goods are given to agent 1, giving the whole cake to agent 2 maximizes the Nash welfare, which is $(0.4 + 0.4) \times 0.002 = 0.0016$.
\item When both indivisible goods are given to agent 2, giving the whole cake to agent 1 maximizes the Nash welfare, which is $0.2 \times (0.499 + 0.499) = 0.1996$.
\item When each agent gets exactly one indivisible good,
% let $x \in [0, 1]$ be the fraction of the cake that is allocated to agent 1.
% One may check that $x = 1$ maximizes the Nash welfare $(0.4 + 0.2 x) \times (0.499 + 0.002 (1-x))$.
in the Nash welfare maximizing allocation, denoted by $\mathcal{A}$, agent 1 receives an indivisible good and the entire cake, and agent 2 receives the other indivisible good.
The Nash welfare of $\mathcal{A}$ is $(0.4 + 0.2) \times 0.499 = 0.2994$.
This is also the overall MNW allocation for this instance.
\end{itemize}
However, allocation $\mathcal{A}$ is not {weak EFM}, because agent 1's bundle {contains} some cake that yields positive value to agent 2, and agent 2 is envious of agent 1.
It is also worth noting that there is a simple envy-free and PO allocation for this instance: each agent gets one indivisible good and one half of the cake, with Nash welfare $(0.4 + 0.1) \times (0.499 + 0.001) = 0.25$.
\end{example}

Note that the compatibility of weak EFM and PO remains an open question even for the special case with indivisible goods and a single homogeneous divisible good (e.g. money), even though this case is well-studied when there is enough money.
% As we mentioned in the section of Related Work, an active line of work studies how to achieve envy-freeness by using money when indivisible goods are to be allocated.
% Thus, perhaps a slightly less ambitious open question would be to study the compatibility of weak EFM and PO for the special case with indivisible goods and a single homogeneous divisible good (like money).}

% We first note that algorithm-wise, finding a PO allocation is a difficult problem even with only divisible goods. Specifically, in cake-cutting, it is known that no finite protocol exists in the Robertson-Webb model that can always produce a PO allocation, even for piecewise uniform valuations~\cite{Ianovski12,KurokawaLaPr13}.

\section{Conclusion and Future Work}
This work is concerned with fair division of a mixture of divisible and indivisible goods.
To this end, we introduce the \longEFM (EFM) fairness notion, which generalizes both EF and EF1 to the mixed goods setting.
We show that an EFM allocation always exists for any number of agents.
We also provide bounded protocols to compute an EFM allocation in special cases, and an $\epsilon$-EFM allocation in the general setting in time poly$(n, m, 1/\epsilon)$.

It remains an important open question whether there exists a bounded, or even finite protocol in the RW model that computes an EFM allocation in the general setting for any number of players.
% Arguably, the perfect allocation oracle cannot be implemented even using a finite number of queries, therefore the question of whether there exists a finite EFM protocol still remains open.
With regard to $\epsilon$-EFM, although our algorithm runs in time poly$(n, m, 1/\epsilon)$, it remains an open question to design an algorithm that runs in time poly$(n, m, \log(1/\epsilon))$.

Besides envy-freeness, one could also generalize other fairness notions to the mixed goods setting.
How well would this notion behave with mixed goods in terms of its existence and approximation?
Overall, we believe that fair division in the mixed goods setting encodes a rich structure and creates a new research direction that deserves to be pursued for future work.

\section*{Acknowledgments}
The authors acknowledge the helpful comments by the reviewers. We are in particular grateful to
an anonymous reviewer for providing Example~\ref{ex:fPO} to us.

This work is supported in part by an RGC grant (HKU 17203717E).

%----------------------------------------------------------------------------------------
%	BIBLIOGRAPHY
%----------------------------------------------------------------------------------------
\bibliographystyle{plainnat}
\bibliography{AIJ}

\end{document}